\RequirePackage{amssymb}
\documentclass[a4paper,UKenglish,cleveref, autoref, thm-restate]{lipics-v2021}

\nolinenumbers

\usepackage{amsmath,amsthm}

\usepackage{thmtools}
\usepackage{thm-restate}

\usepackage[hypcap=false]{caption}

\allowdisplaybreaks

\usepackage{cleveref}

\usepackage{enumerate}

\usepackage{nicefrac}

\usepackage{algorithm}
\usepackage[noend]{algpseudocode}

\usepackage{etoolbox}
\usepackage{tikz}

\usetikzlibrary{tikzmark}
\usetikzlibrary{calc}

\errorcontextlines\maxdimen

\newcommand{\ALGtikzmarkcolor}{black}
\newcommand{\ALGtikzmarkextraindent}{4pt}
\newcommand{\ALGtikzmarkverticaloffsetstart}{-.5ex}
\newcommand{\ALGtikzmarkverticaloffsetend}{-.5ex}
\makeatletter
\newcounter{ALG@tikzmark@tempcnta}

\newcommand\ALG@tikzmark@start{
    \global\let\ALG@tikzmark@last\ALG@tikzmark@starttext
    \expandafter\edef\csname ALG@tikzmark@\theALG@nested\endcsname{\theALG@tikzmark@tempcnta}
    \tikzmark{ALG@tikzmark@start@\csname ALG@tikzmark@\theALG@nested\endcsname}
    \addtocounter{ALG@tikzmark@tempcnta}{1}
}

\def\ALG@tikzmark@starttext{start}
\newcommand\ALG@tikzmark@end{
    \ifx\ALG@tikzmark@last\ALG@tikzmark@starttext
    \else
        \tikzmark{ALG@tikzmark@end@\csname ALG@tikzmark@\theALG@nested\endcsname}%
        \tikz[overlay,remember picture] \draw[\ALGtikzmarkcolor] let \p{S}=($(pic cs:ALG@tikzmark@start@\csname ALG@tikzmark@\theALG@nested\endcsname)+(\ALGtikzmarkextraindent,\ALGtikzmarkverticaloffsetstart)$), \p{E}=($(pic cs:ALG@tikzmark@end@\csname ALG@tikzmark@\theALG@nested\endcsname)+(\ALGtikzmarkextraindent,\ALGtikzmarkverticaloffsetend)$) in (\x{S},\y{S})--(\x{S},\y{E});
    \fi
    \gdef\ALG@tikzmark@last{end}
}

\apptocmd{\ALG@beginblock}{\ALG@tikzmark@start}{}{\errmessage{failed to patch}}
\pretocmd{\ALG@endblock}{\ALG@tikzmark@end}{}{\errmessage{failed to patch}}
\makeatother

\theoremstyle{plain}

\newcommand{\alg}{\text{ALG}}
\newcommand{\opt}{\text{OPT}}

\newcommand{\floor}[1]{\left\lfloor #1 \right\rfloor}
\newcommand{\ceil}[1]{\left\lceil #1 \right\rceil}

\hideLIPIcs

\pdfoutput=1

\title{The \texorpdfstring{\( k \)}{k}-Server with Preferences Problem}

\author{Jannik Castenow}{Heinz Nixdorf Institute \& Department of Computer Science, Paderborn University, F\"urstenallee 11, 33102 Paderborn, Germany}{jannik.castenow@upb.de}{https://orcid.org/0000-0002-8585-4181}{}

\author{Bj\"orn Feldkord}{Heinz Nixdorf Institute \& Department of Computer Science, Paderborn University, F\"urstenallee 11, 33102 Paderborn, Germany}{bjoernf@mail.upb.de}{https://orcid.org/0000-0001-6591-2420}{}

\author{Till Knollmann}{Heinz Nixdorf Institute \& Department of Computer Science, Paderborn University, F\"urstenallee 11, 33102 Paderborn, Germany}{tillk@mail.upb.de}{https://orcid.org/0000-0003-2014-4696}{}

\author{Manuel Malatyali}{Heinz Nixdorf Institute \& Department of Computer Science, Paderborn University, F\"urstenallee 11, 33102 Paderborn, Germany}{malatya@mail.upb.de}{}{}

\author{Friedhelm Meyer auf der Heide}{Heinz Nixdorf Institute \& Department of Computer Science, Paderborn University, F\"urstenallee 11, 33102 Paderborn, Germany}{fmadh@upb.de}{}{}

\authorrunning{Castenow et al.}

\Copyright{Jannik C., Bj\"orn F., Till K. Manuel M., Friedhelm M.a.d.H.}

\ccsdesc[500]{Theory of computation~Online algorithms}
\ccsdesc[500]{Theory of computation~K-server algorithms}
\ccsdesc[300]{Theory of computation~Caching and paging algorithms}

\keywords{K-Server Problem, Heterogeneity, Online Caching}

\relatedversion{A conference version of this paper was accepted at the 34th ACM Symposium on Parallelism in Algorithms and Architectures (SPAA 2022).}

\funding{This work was partially supported by the German Research Foundation (DFG) within the Collaborative Research Centre On-The-Fly Computing (GZ: SFB 901/3) under project number 160364472.}

\begin{document}

\maketitle

\begin{abstract}
    The famous $k$-Server Problem covers plenty of resource allocation scenarios, and several variations have been studied extensively for decades. We present a model generalizing the $k$-Server Problem by \emph{preferences} of the requests, where the servers are not identical and requests can express which specific servers should serve them.
    In our model, requests can either be answered by any server (\emph{general} requests) or by a specific one (\emph{specific} requests). If only general requests appear, the instance is one of the original $k$-Server Problem, and a lower bound for the competitive ratio of $k$ applies. If only specific requests appear, a solution with a competitive ratio of $1$ becomes trivial. We show that if both kinds of requests appear, the lower bound raises to $2k-1$.

    We study deterministic online algorithms and present two algorithms for uniform metrics. The first one has a competitive ratio dependent on the frequency of specific requests. It achieves a worst-case competitive ratio of $3k-2$ while it is optimal when only general requests appear or when specific requests dominate the input sequence. The second has a worst-case competitive ratio of $2k+14$. For the first algorithm, we show a lower bound of $3k-2$, while the second algorithm has a lower bound of $2k-1$ when only general requests appear. The two algorithms differ in only one behavioral rule that significantly influences the competitive ratio. We show that there is a trade-off between performing well against instances of the $k$-Server Problem and mixed instances based on the rule. Additionally, no deterministic online algorithm can be optimal for both kinds of instances simultaneously.

    Regarding non-uniform metrics, we present an adaption of the Double Coverage algorithm for $2$ servers on the line achieving a competitive ratio of $6$, and an adaption of the Work-Function-Algorithm achieving a competitive ratio of $4k$.
\end{abstract}

\section{Introduction}\label{sec:introduction}

Consider the following situation in a distributed system:
There are \( k \) virtual machines, each stored in some node of the system.
Each of them offers an instance of the same service.
Over time, requests for the service appear at the locations of the system.
A request gets served by migrating a virtual machine to the request's location.
The described scenario is an instance of the famous \( k \)-Server Problem.
The virtual machines are the servers, and an algorithm aims to minimize the total movement (migration) cost for serving all requests.
A metric space abstracts the transmission cost in the distributed system and supplies the movement cost.

However, the assumption that every virtual machine offers the same service seems quite restrictive.
It is reasonable to assume that each virtual machine belongs to some provider offering only its services as a bundle.
Presumably, there still exist elementary services that each of the providers offers.
Additionally, each provider offers unique services.
Following this line of thought, each request either only needs elementary services and does not care which virtual machine is provided, or it needs unique services such that specified virtual machines need to be moved.
The new problem statement generalizes the forenamed \( k \)-Server Problem by enabling the requests to express a \emph{preference} on how to be served.
Either a request can be served by any server or by a specific one.
Note, if a request needs multiple specific servers, we can split it without consequences into several requests, one for each specific server.

In this paper, we present a new model -- the \emph{\( k \)-Server with Preferences Problem} -- capturing the above idea and generalizing the \( k \)-Server Problem.
We study our problem in its online version regarding deterministic algorithms mostly on uniform metrics (the paging problem).
We are interested in the \emph{competitive ratio}, a standard analysis technique for the performance of online algorithms.
An algorithm has a competitive ratio of \( c \) if, for any input sequence, its cost is at most \( c \) times the cost of an optimal offline solution on the same input sequence.

At first sight, one might think that it is sufficient to use the following approach:
For all requests that can be answered by any server, apply a standard \( k \)-Server algorithm.
For all requests that require a specific server, move the respective servers to them.
For example, consider a uniform metric space and the Least-Recently-Used (LRU) algorithm, a deterministic marking algorithm achieving an optimal competitive ratio for the \( k \)-Server Problem.
If one applies the strategy from above, there are inputs for which already in a metric space with only three locations and \( k = 3 \) servers the algorithm has an \emph{unbounded competitive ratio} (see \Cref{sec:simple-algorithm:unbounded-competitive}).

Simple adaptations such as the one above fail because our model introduces substantial novelties to the \( k \)-Server Problem.
Due to the requests' preferences, multiple servers possibly must be at identical locations simultaneously.
Furthermore, since the servers each have a unique identity, it no longer suffices to cover the same set of positions as an optimal solution.

\paragraph{Applications for Uniform Metrics}

While our extension of the \( k \)-Server Problem is motivated by general metrics, most of our algorithmic results hold for uniform metrics.
In addition to the encountered analytical depths, this case is motivated by the following applications.

In some practical situations, the migration costs could be dominated by shutting down, wrapping, setting up, and resuming a virtual machine.
In comparison, the data transfer itself is relatively cheap.
Thus, the migration costs can be seen as equal at each location such that the metric space is close to uniform.

For a more paging-related application, imagine a shared memory system where \( k \) caches are tied to a set of processors.
For a fixed processor, the cost to read is the same for all caches while the performance of more complex operations varies due to the caches being wired differently to the processor.
Then it can be required to store a page in a specific cache performing best for the processor executing a complex task on it.
Alternatively, security concerns might require a page to be loaded to a specific cache in some computational context.
The \( k \)-Server with Preferences Problem on uniform metrics provides a first way to model such systems.

\subsection{Our Results}\label{sec:introduction:our-results}

We present a generalization of the \( k \)-Server Problem where requests can (but need not) specify if they want to be served by one specific server.
We call this generalization the \emph{\( k \)-Server with Preferences Problem} and introduce it formally in \Cref{sec:model-and-notation}.
If there are only requests that do not care which server answers them (\emph{general requests}), the problem reduces to the classical \( k \)-Server Problem.
We call such input sequences \emph{pure general inputs}.
If there are only requests that specify which server should serve them (\emph{specific requests}), the problem is trivial as it is clear which server has to move for each request.
Such input sequences are called \emph{pure specific inputs}.
Input sequences consisting of general and specific requests are called \emph{mixed inputs}.

We study the problem regarding deterministic algorithms in the online version.
In that version, requests arrive over time, and each must immediately and irrevocably get served.
The classical \( k \)-Server Problem (pure general inputs) has a lower bound of \( k \) on the competitive ratio, and pure specific inputs pose a competitive ratio of \( 1 \).
We show that mixed inputs pose a higher lower bound of \( 2k-1 \) (see \Cref{sec:lower-bound}).
Additionally, regarding the share \( s \) of specific requests on the total number of requests, we show a lower bound dependent on \( s \) providing a detailed picture of the power of the adversary in our model.

The lower bound already holds in uniform metric spaces, and the design of algorithms with a bounded competitive ratio in these spaces already becomes non-trivial.
Hence, we study online algorithms for the problem in uniform metric spaces.
Our main results are two online algorithms:
\textnormal{\textsc{Conf}} has a competitive ratio dependent on \( s \) as well (see \Cref{sec:k-confident-algorithm}).
It achieves an optimal competitive ratio on pure general (\( s=0 \)) inputs.
For larger \( s \), the ratio increases and overshoots the lower bound up to a competitive ratio of \( 3k - 2 \) (\( s \approx 1/3 \)).
Thereafter, the competitive ratio decreases again with \( s \) until  the optimal ratio for mixed inputs is reached (\( s=\nicefrac{1}{2} \)).
Thereafter, the ratio drops rapidly (and \emph{independently of \( k \)}).
More specifically, the algorithm's competitive ratio matches the lower bound for \( s>\nicefrac{1}{2} \) until it reaches \( 1 \) for pure specific inputs (\( s=1 \)).
The bound of \textnormal{\textsc{Conf}} nicely shows how the adversary is significantly stronger for \( s \leq \nicefrac{1}{2} \) and loses its power as soon as specific requests dominate the input.
Our second algorithm, \textnormal{\textsc{Def}}, is designed to optimize the worst-case competitive ratio.
It achieves a competitive ratio of \( 2k + 14 \) on all inputs, but its competitive ratio for pure general inputs is inherently lower bounded by \( 2k - 1 \) (see \Cref{sec:k-defensive-algorithm}).
As a side-note, \textnormal{\textsc{Def}} can be straightforwardly improved such that, for \( s > \nicefrac{1}{2} \) it also achieves an optimal competitive ratio.
Consider \Cref{figure:competitive-ratio-plot} for an overview of the competitive ratios of our algorithms dependent on \( s \).
Our results indicate a trade-off between performing well on mixed inputs and performing well on pure general inputs.
In fact, we show that \emph{no} algorithm can achieve an optimal competitive ratio for all input types (\Cref{th:no-algorithm-can-be-super} in \Cref{sec:defensive-confident}).

It turns out that there is \emph{one critical behavior} an algorithm can implement for each server that significantly influences the competitive ratio.
We call this behavior \emph{acting defensively}.
If an algorithm acts defensively for \( j \), it keeps track of the last known optimal position \( p^{*}(j) \) of \( j \) (given by the initial configuration or specific requests) and, if a general request appears on \( p^{*}(j) \), always moves \( j \) to the request.
Simple adapted algorithms from the original \( k \)-Server Problem (e.g., LRU, FIFO) usually do not always act defensively for all servers (as they move them in an order independent of the last known optimal positions).
If an algorithm acts not always defensively for some server, we say, it acts \emph{confidently} for that server.
We present the precise definition, more intuitions, and the precise influence of acting defensively/confidently on the competitive ratio in \Cref{sec:defensive-confident}.

As mentioned before, getting to the bound of \( k \) for pure general inputs yields a sub-optimal worst-case bound for mixed inputs and vice versa.
We study this trade-off by considering defensive and confident algorithms.
Our main findings here are the following:
For an extreme class of algorithms that \emph{always} act confidently for \( \ell \) servers (strictly-\( \ell \)-confident algorithms), we show that there is an even higher worst-case lower bound of \( 2k + \ell - 2 \) on mixed inputs.
Our algorithm \textnormal{\textsc{Conf}} (\Cref{sec:k-confident-algorithm}) does not fall into this class, as it may not always act confidently for all servers (it is a \( k \)-confident algorithm).
However, by a similar lower bound, it suffers the same higher worst-case bound of \( 3k-2 \) (see \Cref{sec:lower-bound-conf}).
Therefore, we strongly believe that a similar lower bound exists for all \( \ell \)-confident algorithms.
The advantage of confident algorithms is that they can achieve an optimal competitive ratio of \( k \) on pure general inputs.
To get closer to the worst-case lower bound of \( 2k-1 \), we shift our focus to algorithms that always act defensively for \( \ell > 0 \) servers (\( \ell \)-defensive algorithms).
We show that any such algorithm also has a lower bound of \( k + \ell - 1 \) on the competitive ratio for pure general inputs.
Our algorithm \textnormal{\textsc{Def}} is a \( k \)-defensive algorithm (\Cref{sec:k-defensive-algorithm}).
\\

\begin{minipage}[c]{0.46\textwidth}
    \centering
    \includegraphics[page=1, width=\textwidth, clip=true, trim=0.35cm 0.5cm 0.7cm 0.5cm]{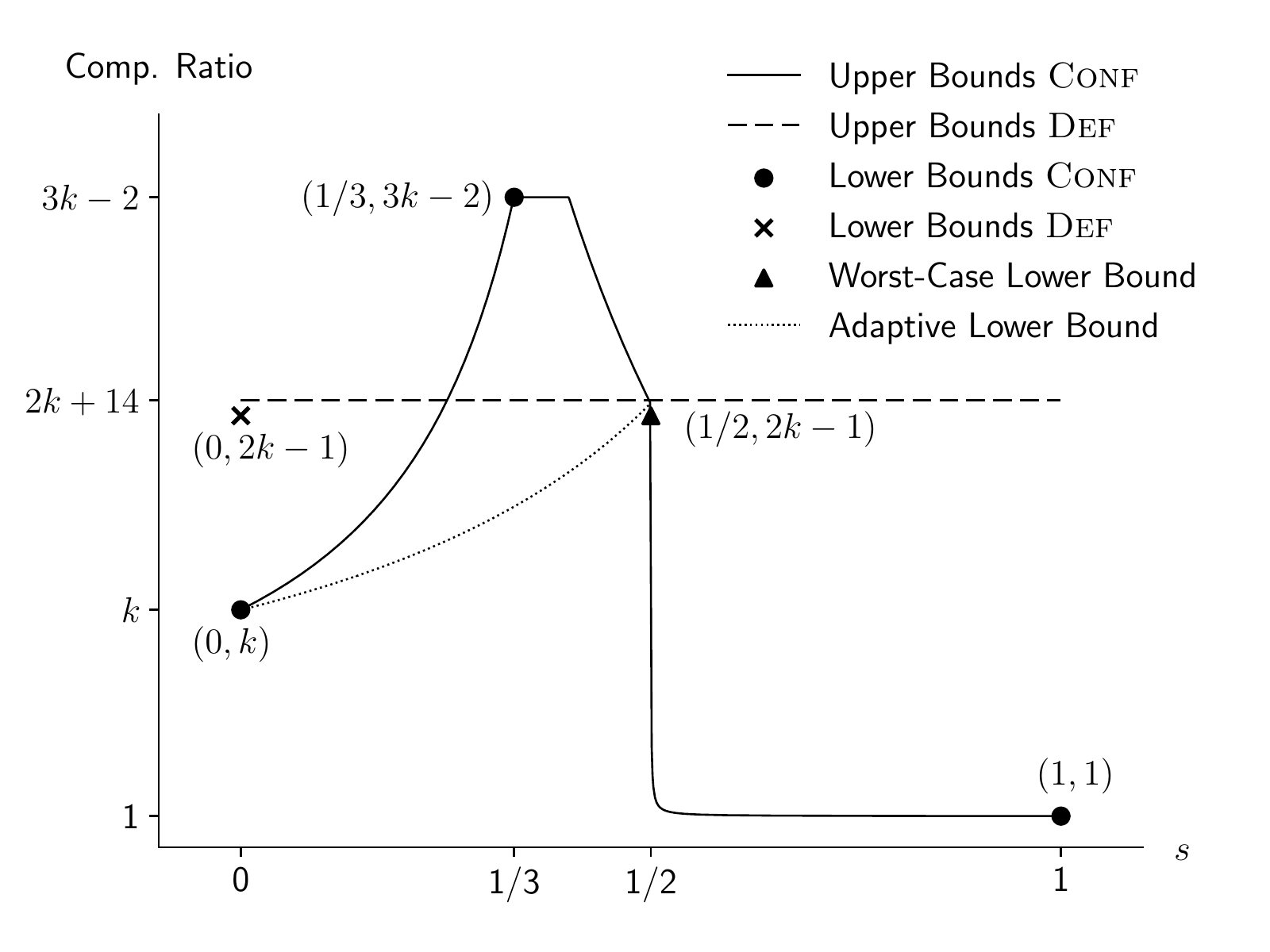}
\end{minipage}
\begin{minipage}[c]{0.02\textwidth}
    \hfill
\end{minipage}
\begin{minipage}[c]{0.46\textwidth}
    \centering
    \captionof{figure}{The competitive ratios of our algorithms depend on \( s \), the number of specific requests divided by the total number of requests (regarding only requests that require a movement by the online algorithm).
    \textnormal{\textsc{Conf}} acts confidently on all servers, achieves an optimal competitive ratio in instances of the \( k \)-Server Problem (\( s=0 \)) and when specific requests dominate the input sequence (\( s \geq 1/2 \)), and a worst-case competitive ratio of \( 3k - 2 \) (\( s \approx 1/3 \)).
    \textnormal{\textsc{Def}} acts defensively on all servers and achieves a worst-case competitive ratio of \( 2k + 14 \).
    The price for this improvement is that the competitive ratio for instances of the \( k \)-Server Problem cannot be better than \( 2k - 1 \) (\( s=0 \)).}
\label{figure:competitive-ratio-plot}
\end{minipage}
\noindent\\

Our algorithms can also be combined because they only differ in acting defensively or confidently for each server.
A combined algorithm that is \( \ell \)-confident and \( (k - \ell) \)-defensive can be analyzed by applying the analysis of \textnormal{\textsc{Conf}} for all servers that act confidently and applying the analysis of \textnormal{\textsc{Def}} for all servers that act defensively.
The resulting algorithm then achieves a competitive ratio of at most \( 2k - \ell + 14 \) on pure general inputs and at most \( 2k + \ell + 14 \) on mixed inputs.
Intuitively, each server for which the algorithm acts defensively increases the competitive ratio on pure general inputs and decreases it on mixed inputs.
\bigskip

In \Cref{sec:general}, we extend our results by considering non-uniform metrics.
We present an algorithm for \( 2 \) servers on the real line which is based on the Double Coverage algorithm.
This algorithm achieves a competitive ratio of \( 6 \).
Designing algorithms for more than \( 2 \) servers is surprisingly more difficult as the algorithm has to decide on how to cover \( k \) previously known locations with \( k-1 \) remaining servers when a server gets forced away from its location due to a specific request.
Further, we present a simple adaption of the Work-Function-Algorithm that achieves a competitive ratio of \( 4k \).
It remains open if and how one can achieve a competitive ratio closer to the lower bound of \( 2k-1 \) in non-uniform metrics.

\subsection{Related Work}\label{sec:introduction:related-work}

The \( k \)-Server Problem was introduced in 1988~\cite{manasse_competitive_1988} shortly after the introduction of the competitive ratio (amortized efficiency in \cite{sleator_amortized_1985}), which evolved into a standard analysis technique for online algorithms.
The problem is a special case of Metrical Task Systems \cite{borodin_optimal_1992}.
In \cite{manasse_competitive_1988}, the authors showed a lower bound of \( k \) for deterministic algorithms and stated the famous \( k \)-Server Conjecture, claiming that there is an algorithm with a competitive ratio of \( k \) for every metric space.
Various work tackled the conjecture, but until today it was only proven for several restricted cases (restricted \( k \)) and simple metrics (\cite{chrobak_new_1990,chrobak_optimal_1991,koutsoupias_2-evader_1996}).
Notably, the Double Coverage algorithm~\cite{chrobak_new_1990} achieves an optimal competitive ratio of \( k \) on the real line.
The most general result as of today is the Work Function Algorithm presented in \cite{koutsoupias_k-server_1995} that achieves a competitive ratio of \( 2k - 1 \) on general metric spaces and an optimal competitive ratio on several specific metrics.
Besides the work toward the conjecture, several variations of the original model were introduced.
Among others, there is the generalized \( k \)-Server Problem~\cite{koutsoupias_cnn_2000}, where each server lies in its own metric space, a model where requests can get served at later time steps for some penalty payment \cite{azar_online_2021}, and a model where requests can be rejected \cite{bittner_online_2014}.
With a focus on the analysis, recently, models with advice were considered, where restricted knowledge on future requests is granted \cite{bockenhauer_advice_2011}.
For a more detailed overview of the history of the \( k \)-Server Problem, we refer to the excellent survey by Koutsoupias~\cite{koutsoupias_k-server_2009}.

The \( k \)-Server Problem generalizes paging (caching), i.e., paging is the \( k \)-Server Problem on a uniform metric.
Later in this paper, we focus on algorithms for this special case of the \( k \)-Server Problem.
For paging, there is a class of optimal deterministic algorithms called \emph{marking algorithms} (see \cite[pp.~752-758]{kleinberg_algorithm_2014} for an explanation).
The general idea of marking algorithms is to operate in phases such that the optimal solution pays at least a cost of \( 1 \) per phase, while the algorithm moves each server once per phase and remembers already used servers by marking them.
They include well-established algorithms such as the Least-Recently-Used (LRU) and the First-In-First-Out (FIFO) policy.
While our algorithms for uniform metrics rely on the marking technique, we want to stress again (as mentioned in the introduction) that simple adaptions of marking algorithms end up with an unbounded competitive ratio.

\subsection{A simple LRU adaption has an unbounded Competitive Ratio}\label{sec:simple-algorithm:unbounded-competitive}

In the following section, we show how a simple adaption of the Least-Recently-Used (LRU) policy fails in our model on uniform metrics.
More specifically, the presented algorithm yields an unbounded competitive ratio already on a tiny and simple metric.
Therefore, we require more involved algorithms.
In the following, we use the notation introduced in \Cref{sec:model-and-notation}.
\bigskip

\textbf{Example Algorithm.}
Use the least recently used server for a general request.
For a specific request, use the requested server.
This server is then also counted as the most recently used one.
\bigskip

For pure general inputs, i.e., instances of the $k$-Server Problem, the LRU approach is optimal.
Unfortunately, we can show the following for mixed inputs:

\begin{theorem}\label{th:simple-algorithm-not-competitive}
    The Example Algorithm has an unbounded competitive ratio even for \( k = 3 \) servers on a uniform metric with \( k \) locations.
\end{theorem}

\begin{proof}
    Consider the locations \( v_{1}, v_{2}, v_{3} \), \alg{}'s servers \( a_{1}, a_{2}, a_{3} \) and \opt{}'s servers \( o_{1}, o_{2}, o_{3} \).
    Let initially \( p(a_{i}) = p(o_{i}) = v_{i} \) for all servers.
    We assume that the algorithm starts to move its \( a_{i} \) in the order of the indices.
    \opt{} moves each server once such that \( p(o_{1}) = v_{2} \), \( p(o_{2}) = v_{3} \) and \( p(o_{3}) = v_{1} \) for a cost of \( 3 \).
    Afterward, \opt{} never moves its servers.

    Start the sequence by requesting server \( 1 \) specifically at \( v_{2} \).
    Now, the algorithm has its servers \( a_{1}, a_{2} \) at \( v_{2} \) and \( a_{3} \) at \( v_{3} \), while the order of movements for the servers is \( a_{2}, a_{3}, a_{1} \).
    Next, we show how to construct a sequence of requests such that the algorithm has a cost and ends up in the same situation again.
    First, apply a general request on \( v_{1} \) causing \( a_{2} \) to move there.
    Then, request server \( 3 \) specifically at \( v_{1} \), such that \( a_{3} \) moves.
    Do a general request on \( v_{3} \) to move \( a_{1} \) there and afterward, a general request on \( v_{2} \) such that \( a_{2} \) moves there.
    Specifically request server \( 1 \) on \( v_{2} \) to move \( a_{1} \).
    Finally, make a general request on \( v_{3} \) such that the algorithm moves \( a_{3} \) there.
    Repeating this sequence infinitely long yields the theorem.
\end{proof}

\Cref{th:simple-algorithm-not-competitive} already shows us that there are multiple pitfalls in our extended model.
It seems disadvantageous to cover a position with two servers if not needed due to specific requests.
Also, the algorithm might make a mistake when moving a server if that server is immediately specifically requested at its previous position.
It seems to be a good idea to keep servers at the position where they were lastly specifically requested.
However, we have to move a server away from this position eventually.

\section{Model and Notation}\label{sec:model-and-notation}

We consider a generalization of the online \( k \)-Server Problem on a metric space \( M \).

\paragraph{The \( k \)-Server Problem.}
The \( k \)-Server Problem is defined as follows:
There are \( k \) servers \( 1,\dots, k \) located at the points of the metric space.
Their initial position is predetermined.
Denote by \( K \) the set of all servers.
An input sequence consists of a sequence of requests \( r_{1}, \dots, r_{n} \) arriving in time steps \( 1,\dots, n \).
Each request \( r \) appears at a location of \( M \) and must be answered by moving a server onto \( r \).
The cost of an algorithm for this problem is the total distance moved by all servers.
An algorithm aims at minimizing the total cost.
In the online version, the requests are revealed one by one and must immediately be answered, i.e., request \( r_{i} \) reveals after the algorithm answered request \( r_{i-1} \).

\paragraph{The \( k \)-Server with Preferences Problem.}
We generalize the \( k \)-Server Problem to the \( k \)-Server with Preferences Problem as follows:
Each request \( r \) is either \emph{general} or \emph{specific}.
A general request is answered by moving \emph{any} server on \( r \).
A specific request asks for one server \( s \in K \) and is answered by moving \( s \) onto \( r \).

Note that solutions to our problem also solve a more general problem where each request \( r \) asks for a set of servers \( s(r) \subseteq K \) that all have to move onto it.
For this, simply replace any specific request \( r \) for \( s(r) \) by one specific request for each server \( j \in s(r) \) in the input sequence.
Any solution to the transformed input also is a solution for the more general problem with the original input.

To remove the dependency on the initial configuration in the competitive ratio, we assume for all our algorithms that initially, every server \( j \) of the online algorithm is at the same location as server \( j \) of the optimal solution.
This assumption can be dropped while adding only an additive term independent of the optimal solution in the competitive ratio.

\paragraph{Additional Notation.}
When an input sequence contains only general requests, the \( k \)-Server with Preferences Problem becomes equal to the original \( k \)-Server Problem.
We call such input sequences \emph{pure general inputs}, input sequences containing both general and specific requests \emph{mixed inputs}, and input sequences solely consisting of specific requests \emph{pure specific inputs}.
For any server \( j \), we refer to its current position by \( p(j) \) and to the last position where it was specifically requested by \( p^{*}(j) \).
\( p^{*}(j) \) is initially set to \( j \)'s initial position.
If we have a set of servers \( U \), we denote by \( p(U) \) the set of positions of servers of \( U \).
We denote by \opt{} the optimal offline solution that knows all requests beforehand.
In our analyses, we need a precise understanding of time.
For this, we define \( t \) to be the time-step at which the \( t \)-th request appears.
We denote by \emph{right before} \( t \), the point in time at the beginning of \( t \) before the algorithm and the optimal solution have moved their servers and answered the request, while the request is already revealed.
We denote by \emph{right after} \( t \), the point in time at the end of \( t \) after the algorithm and \opt{} moved their servers.
Observe that considering time \( t \), right after \( t \) is before right before \( t+1 \).

\section{The Lower Bound is \texorpdfstring{\( 2k - 1 \)}{2k-1}}\label{sec:lower-bound}

In the following section, we show that the lower bound for the \( k \)-Server with Preferences Problem is significantly higher than the lower bound of \( k \) of the classical \( k \)-Server Problem.

\begin{restatable}{thm}{GeneralLowerBound}\label{th:general-lower-bound}
    Every deterministic online algorithm for the \( k \)-Server with Preferences Problem has a competitive ratio of at least $2k-1$ even in a uniform metric space with \( k+1 \) locations.
\end{restatable}

\begin{proof}
    Consider the uniform metric with locations \( v_{1}, \dots, v_{k+1} \).
    Let \( a_{1}, \dots, a_{k} \) be the servers of an online algorithm and \( o_{1}, \dots, o_{k} \) be the servers of an optimal solution.
    Assume that initially, \( p(a_{i}) = p(o_{i}) = v_{i} \) for all \( 1 \leq i \leq k \).

    The request sequence starts with a general request on \( v_{k+1} \) and proceeds to pose general requests at the currently unoccupied location.
    Since the online algorithm is deterministic, there is a permutation \( \pi \) of \( 1,\dots, k \) such that the algorithm moves its servers for the first time in the order given by \( \pi \), i.e., the first time \( a_{\pi(i)+1} \) moves, \( a_{\pi(i)} \) has already moved once.
    Rename each \( i \) to \( \pi(i) \).
    Now, the algorithm moves its servers in the order given by \( 1, \dots, k \), i.e., the first time \( a_{i+1} \) moves, \( a_{i} \) has already moved once.
    Now, we build the sequence as described above, but we do not pose any request for \( v_{k} \).
    The online algorithm will move \( a_{k} \) the last, while the optimal solution will only move \( o_{k} \).
    Note how we can assume that the algorithm moves each server eventually because else, \( a_{k} \) is never moved.
    Then, the competitive ratio is unbounded, because the algorithm never converges towards the optimal locations without using \( a_{k} \).
    Separate the sequence into phases:
    We say phase \( i \) starts when \( a_{i} \) is moved for the first time.
    During each phase \( i \) only the servers \( a_{1}, \dots, a_{i} \) move and the only locations that can become unoccupied are \( v_{1},\dots, v_{i} \) and \( v_{k+1} \).
    The last phase ends when the online algorithm occupies \( v_{1},\dots,v_{k-1} \) and \( v_{k+1} \), i.e., the same locations as the optimal solution.
    Until then, by the definition of the phases, the online algorithm must have moved each server at least once.
    For any server \( a_{i} \), for \( i \leq k-1 \), at the end of the last phase it holds that either \( a_{i} \) was moved at least twice or \( a_{i} \) is not located at \( v_{i} \).

    Next, we issue a specific request for each of \( a_{1}, \dots, a_{k-1} \) at their initial positions.
    With the exception of \( a_{k} \), each server \( a_{i} \) that is not at \( v_{i} \) must be moved there.
    Hence, every server except \( a_{k} \) moved twice, and \( a_{k} \) moved once yielding a total cost of \( 2k-1 \).
    The optimal solution answers the total sequence by moving \( o_{k} \) to \( v_{k+1} \) for a cost of \( 1 \).
    Now, we are in the same configuration as in the beginning by renaming the locations.
\end{proof}

For a depiction of the lower bound sequence of \Cref{th:general-lower-bound}, consider Figure~\ref{figure:lower-bound-sequence}.

\begin{minipage}[c]{0.46\textwidth}
    \centering
    \includegraphics[page=4, width=\textwidth, clip=true, trim= 0cm 8cm 18cm 0cm]{figures/figures.pdf}
\end{minipage}
\begin{minipage}[c]{0.02\textwidth}
    \hfill
\end{minipage}
\begin{minipage}[c]{0.46\textwidth}
    \centering
    \captionof{figure}{The lower bound sequence of \Cref{th:general-lower-bound} is divided into two phases.
    In the first one, always the currently unoccupied node (except for \( v_{k} \)) is requested.
    Eventually, the online algorithm covers all nodes except \( v_{k} \).
    Then the second phase forces every server except for \( a_{k} \) back to its initial position using specific requests.
    The optimal solution solves the entire sequence by moving \( o_{k} \) to \( v_{k+1} \).}
    \label{figure:lower-bound-sequence}
\end{minipage}
\noindent\\

The bound demonstrates that an algorithm not only has to determine the optimal positions of the servers but also the precise mapping of specific servers to positions.
Therefore, the lower bound for mixed inputs is by \( k-1 \) higher than the one for pure general inputs.

We remark that a lower bound to our model was also presented in \cite{haneyAlgorithmsNetworksUncertainty2019}.
Our lower bound is similar in spirit but different in structure.
Its structure allows us to extend it for classes of algorithms presented in \Cref{sec:defensive-confident}.

Next, we show an adaptive lower bound that depends on the share \( s \) of relevant specific requests to complement the bound above.
For this, we need the lemma below.

\begin{restatable}{lem}{helpfulLemmaRationalNumber}\label{le:helpful-lemma-rational-numbers}
    Let \( x \in \mathbb{Q} \) be a rational number greater than one that is not a natural number.
    Then there are natural numbers \( a,b \in \mathbb{N} \) with \( a,b > 0 \) such that \( x = \frac{a \cdot \floor{x} + b \cdot \ceil{x}}{a+b} \).
\end{restatable}

\begin{proof}
    Since \( x \) is not natural but rational, it holds that \( \ceil{x} = \floor{x} + 1 \) and \( x = \floor{x} + \varepsilon \) for \( \varepsilon \in \mathbb{Q} \) and \( 0 < \varepsilon < 1 \).
    Set for any two numbers \( a,b \):
    \begin{align*}
        \floor{x} + \varepsilon = x = \frac{a \cdot \floor{x} + b \cdot \ceil{x}}{a + b} = \frac{a \cdot \floor{x} + b \cdot (1 + \floor{x})}{a + b}.
    \end{align*}
    Simplifying and solving for \( a \) yields \( a = b \left(\frac{1}{\varepsilon} - 1 \right) \).
    Since \( \varepsilon \in \mathbb{Q} \) and \( \varepsilon < 1 \), \( \frac{1}{\varepsilon} \) is rational and greater than one such that \( \left(\frac{1}{\varepsilon}-1\right) \) is rational and greater than zero.
    Thus, it can be expressed as a fraction \( \frac{c}{d} \) of two natural numbers \( c \) and \( d \) greater than zero.
    Set \( b = d \).
    Then \( a = c \) and both \( a \) and \( b \) are natural numbers greater than zero which shows the lemma.
\end{proof}

\begin{restatable}{thm}{AdaptiveLowerBound}\label{th:adaptive-lower-bound}
    Consider any deterministic online algorithm for the \( k \)-Server with Preferences Problem.
    Let \( g, f \) be the number of general/specific requests that require the algorithm to move.
    Let \( s = \nicefrac{f\,}{\,g + f} \) be the share of specific requests on the total number of requests that require a movement by the algorithm.
    Already in a uniform metric space with \( k+1 \) locations it holds:
    The algorithm has a competitive ratio of at least \( (1 + \frac{s}{1-s})\,k \) if \( s \leq \frac{k-1}{2k-1} \), a competitive ratio of at least \( 2k - 1 \) if \( \frac{k-1}{2k-1} < s < \frac{k}{2k-1} \), and a competitive ratio of at least \( \frac{1}{2\,s - 1} \) if \( s \geq \frac{k}{2k-1} \).
\end{restatable}

\begin{proof}
    Consider the uniform metric with locations \( v_{1}, \dots, v_{k+1} \).
    Let \( a_{1}, \dots, a_{k} \) be the servers of an online algorithm and \( o_{1}, \dots, o_{k} \) be the servers of an optimal solution.
    Assume that initially, \( p(a_{i}) = p(o_{i}) = v_{i} \) for all \( 1 \leq i \leq k \).
    We further assume that the algorithm moves \emph{lazy}, i.e., it only moves a server if there is an unserved request at the destination.
    As it was pointed out in \cite{koutsoupias_k-server_2009}, every algorithm can be turned \emph{lazy} without disadvantage.
    This assumption allows us to reasonably capture \( s \) since it depends on the requests requiring a movement of the algorithm.

    Assume \( s \leq \frac{k-1}{2k-1} < \frac{1}{2} \).
    As in the proof \Cref{th:general-lower-bound}, assume that the \( i \)'s are renamed such that the algorithm starts to move its \( a_{i} \) in the order of the indices during the phase below.
    The sequence is a modification of the first phase of the proof of \Cref{th:general-lower-bound}:
    First, issue a general request on \( v_{k+1} \).
    For \( x < k \) steps do the following (\( x \) is determined later):
    After the algorithm moved a server \( a_{i} \) to \( v_{k+1} \), issue a specific request for \( a_{i} \) at \( v_{i} \) (which is unoccupied).
    Issue another general request on \( v_{k+1} \) (which is unoccupied again).
    After \( x \) many steps, always pose a general request on the currently unoccupied position until the algorithm covers all locations except \( v_{k} \).
    It then holds that \( x \) many servers have been moved twice (to \( v_{k+1} \) and back to their initial position) while the remaining \( k - x \) servers have been moved at least once (similar to the proof of \Cref{th:general-lower-bound}, else, the algorithm has unbounded competitive ratio).
    Observe that all requests are to unoccupied locations such that all of them require a movement while the number of specific requests is \( x \) while the total number of requests is \( k + x \).
    Overall, the cost of the algorithm is at least \( k + x \).
    The optimal solution solves the entire sequence by moving \( o_{k} \) to \( v_{k+1} \) for a cost of \( 1 \) and the competitive ratio is lower bounded by \( k + x \).
    Observe that \( x \) must be a natural number.
    If \( \frac{s}{1-s} \, k \) is a natural number, we can simply set \( x = \frac{s}{1-s} \, k \) such that \( \frac{x}{k + x} = s \) and the competitive ratio is at least \( (1 + \frac{s}{1-s}) \, k \).
    Else, observe that the final configuration is equivalent to the initial one in the input sequence above.
    Hence, the sequence can be repeated arbitrarily often.
    By \Cref{le:helpful-lemma-rational-numbers} we know that there are natural numbers \( a,b \) greater than zero such that \( \frac{s}{1-s} \, k = \frac{a \floor{\frac{s}{1-s} \, k} + b \ceil{\frac{s}{1-s} \, k}}{ a + b } \).
    For \( a \) phases, set \( x = \floor{\frac{s}{1-s} \, k} \) and for \( b \) phases, use \( x = \ceil{\frac{s}{1-s} \, k} \).
    Then the ratio between the number of specific requests and the total number of requests is:
    \begin{align*}
        \frac{a \floor{\frac{s}{1-s} \, k} + b \ceil{\frac{s}{1-s} \, k}}{a \floor{\frac{s}{1-s} \, k} + b \ceil{\frac{s}{1-s} \, k} + (a+b)\, k}
        = \frac{(a+b) \, \frac{s}{1-s} \, k}{ (a+b) \, \frac{s}{1-s} \, k + (a+b)\,k}
        = \ldots
        = s.
    \end{align*}
    The competitive ratio is at least:
    \begin{align*}
        \frac{a \left(k + \floor{\frac{s}{1-s} k}\right) + b \left(k + \ceil{\frac{s}{1-s} k}\right)}{a + b}
        = \frac{a k + b k + a \floor{\frac{s}{1-s} k} + b \ceil{\frac{s}{1-s} k}}{a + b}
        = \left(1 + \frac{s}{1-s} \right) k.
    \end{align*}
    Observe for any case that \( s \leq \frac{k-1}{2k-1} \) ensures that:
    \begin{align*}
        \frac{s}{1-s} \, k
        \leq \ceil{\frac{s}{1-s} \, k}
        \leq \ceil{\frac{\frac{k-1}{2k-1}\cdot k}{1-\frac{k-1}{2k-1}}}
        \leq \ldots
        \leq \ceil{k-1}
        \leq k-1.
    \end{align*}
    Therefore, truly in each sequence \( x < k \).
    Additionally, observe that for \( s = \frac{k-1}{2k-1} \) the competitive ratio is \( 2k-1 \).

    Next, consider \( s \geq \frac{k}{2k-1} > \frac{1}{2} \).
    As in the proof \Cref{th:general-lower-bound}, assume that the \( i \)'s are renamed such that the algorithm starts to move its \( a_{i} \) in the order of the indices during the sequence below.
    Similar to before, the sequence is a modification of the first phase of the proof of \Cref{th:general-lower-bound}.
    However, in comparison to the case of \( s < \frac{1}{2} \), there are insufficiently many general requests to move every server.
    This is compensated by specific requests.
    In the first phase, do for \( x < k \) steps the following (\( x \) is determined later):
    Issue a general request on \( v_{k+1} \).
    Assume the algorithm moves server \( a_{i} \).
    Issue a specific request for \( a_{i} \) on \( v_{i} \).
    Since the algorithm moves the servers in order and lazy, it incurs a cost of \( 2 \, x \) for the first phase and moves only servers \( a_{i} \) with \( i \leq x \).
    In phase two, consider \( k - x - 1 \) servers with an index between \( x \) and \( k \) (so, server \( k \) is not considered).
    Find a permutation \( \pi \) for the initial locations of the above servers such that \( v_{i} \neq \pi(v_{i}) \) for all of them.
    Pose a specific request for each of these servers \( i \) at \( \pi(v_{i}) \).
    Finally, phase three is a single specific request for \( k \) at \( v_{k+1} \).
    Observe that all requests require a movement by the algorithm and the number of specific requests is \( k \) while the total number of requests is \( k + x \).
    The total cost of the algorithm is at least \( k + x \).
    The optimal solution only moves the servers of phases two and three and has a cost of at most \( k - x \).
    Simplifying the ratio gives a lower bound of \( \frac{k + x}{k - x} \) on the competitive ratio.
    Again, \( x \) must be a natural number.
    If \( \frac{1-s}{s} \, k \) is a natural number, simply set \( x = \frac{1-s}{s} \, k \) such that \( \frac{k}{k+x} = s \) and the lower bound is \( \frac{1}{2\,s - 1} \).
    Else, observe that the final configuration is equivalent to the initial one in the input sequence above.
    Hence, the sequence can be repeated arbitrarily often.
    By \Cref{le:helpful-lemma-rational-numbers}, we know that there are natural numbers \( a,b \) greater than zero such that \( \frac{1-s}{s} \, k = \frac{a \floor{\frac{1-s}{s} \, k} + b \ceil{\frac{1-s}{s} \, k}}{a+b} \).
    For \( a \) phases, set \( x = \floor{\frac{1-s}{s} \, k} \) and for \( b \) phases, use \( x = \ceil{\frac{1-s}{s} \, k} \).
    Then the ratio between the number of specific requests and the total number of requests is:
    \begin{align*}
        \frac{a k + b k}{a k + b k + a \floor{\frac{1-s}{s} \, k} + b \ceil{\frac{1-s}{s} \, k}}
        = \frac{(a+b)\,k}{ (a+b)\,k + (a+b) \, \frac{1-s}{s} \, k}
        = \ldots
        = s.
    \end{align*}
    The competitive ratio is at least:
    \begin{align*}
        \frac{ak + a \floor{\frac{1-s}{s} \, k} + b k + b \ceil{\frac{1-s}{s} \, k}}{a k - a \floor{\frac{1-s}{s} \, k} + b k - b \ceil{\frac{1-s}{s} \, k}}
        = \frac{(a+b) k + (a+b)\, \frac{1-s}{s} \, k}{(a+b)k - (a+b)\,\frac{1-s}{s} \, k}
        = \ldots
        = \frac{1}{2\,s - 1}.
    \end{align*}
    Observe that in any case \( s \geq \frac{k}{2k-1} \) ensures that:
    \begin{align*}
        \frac{1-s}{s} \, k
        \leq \ceil{\frac{1-s}{s} \, k }
        \leq \ceil{\frac{\ceil{\left(1-\frac{k}{2k-1}\right)\,k}}{\frac{k}{2k-1}}}
        \leq \ldots
        \leq \ceil{k-1}
        \leq k-1.
    \end{align*}
    Therefore, truly in each sequence \( x < k \).
    Additionally, observe that for \( s = \frac{k}{2k-1} \) the competitive ratio is \( 2k-1 \).

    Consider the case of \( \frac{k-1}{2k-1} < s < \frac{k}{2k-1} \) where \( s \approx \frac{1}{2} \).
    First observe that the above bounds on \( s \) imply that (1) \(  k - 1 < s (2k-1) < k \).
    As observed in the first part of the proof, for \( s = \frac{k-1}{2k-1} \) the first input sequence gives a competitive ratio of at least \( 2k-1 \).
    Also, for \( s = \frac{k}{2k-1} \), the second input sequence yields a competitive ratio of at least \( 2k-1 \).
    In both cases, the configuration after the sequence is the same as the initial one and we can repeat each sequence arbitrarily often.
    By (1) \( s  (2k-1) \) is rational but not natural and, therefore, by \Cref{le:helpful-lemma-rational-numbers}, there are natural numbers \( a,b \) greater than zero such that \( s  (2k-1) = \frac{a\floor{s  (2k-1)} + b\ceil{s  (2k-1)}}{a+b} \).
    Observe that \( \floor{s (2k-1)} = k-1 \) and \( \ceil{s (2k-1)} = k \).
    Repeat \( a \) times the sequence for \( s = \frac{k-1}{2k-1} \) and \( b \) times the sequence for \( s = \frac{k}{2k-1} \).
    Then, the ratio between the number of specific requests and the total number of requests is
    \begin{align*}
        \frac{a\frac{k-1}{2k-1} + b\frac{k}{2k-1}}{a+b}
        &= \frac{a\frac{\floor{s (2k-1)}}{2k-1} + b\frac{\ceil{s (2k-1)}}{2k-1}}{a+b}
        = \frac{a\floor{s (2k-1)} + b\ceil{s (2k-1)}}{a+b}\cdot\frac{1}{2k-1}
        \\
        &= \frac{s (2k-1)}{2k-1}
        = s.
    \end{align*}
    Finally, the competitive ratio is \( 2k-1 \).
\end{proof}

\section{Defensive and Confident Algorithms}\label{sec:defensive-confident}

Next, we identify a key behavior of an algorithm for our setting.
For this, consider a single server \( j \) of our algorithm.
Initially, we know that \( j \) is at the same position as it is in the optimal solution.
Additionally, after every specific request \( r \) for \( j \), we know that \( j \) has to be located at \( r \) in any solution.

Consider now the following scenario:
\( j \) is located at \( p^{*}(j) \) and moves due to some general request to the location \( p \).
If now a general request on \( p^{*}(j) \) appears, the online algorithm could have made a mistake by moving \( j \) from \( p^{*}(j) \).
Maybe the optimal solution keeps \( j \) on \( p^{*}(j) \) and moved some other server to \( p \).
If that is the case, the algorithm should move \( j \) back to \( p^{*}(j) \) because else, it moves some other server \( s \) to \( p^{*}(j) \) which is definitively a mistake because the optimal solution did not have to move \( s \) to \( p^{*}(j) \).
On the other hand, if the optimal solution did not keep \( j \) on \( p^{*}(j) \), moving \( j \) back to \( p^{*}(j) \) itself is a mistake and the algorithm should place some other server on \( p^{*}(j) \).
Whether moving \( j \) back to \( p^{*}(j) \) is a good decision is only (if ever) revealed when the next specific request for \( j \) appears.
Then the algorithm can determine if \( j \) must have been moved by the optimal solution.

In other words, in such a situation, the algorithm has to decide the movement of \( j \) based on a guess of what the optimal solution does.
Either the algorithm assumes that \( j \) is kept on \( p^{*}(j) \) in the optimal solution, or it does not.
If the algorithm decides to move \( j \) back to \( p^{*}(j) \), we say it \emph{acts defensively} for \( j \), because it assumes it made a mistake by moving \( j \) away earlier.
Else, it acts \emph{confidently} because it assumes it did not make a mistake.
We will see that whether to act defensively is a critical decision that significantly influences the competitive ratio an algorithm can achieve.
Next, we categorize all online algorithms into two classes based on whether they act defensively (\Cref{def:defensive-confident}).

\begin{definition}[Defensive/Confident]\label{def:defensive-confident}
    An algorithm acts defensively for \( j \), if, when a general request appears on \( p^{*}(j) \) that requires a movement, it moves \( j \) to the request.
    A deterministic algorithm is \( \ell \)-\emph{defensive}, if for \( \ell \) servers, it \emph{always} acts defensively.
    An algorithm is \( \ell \)-\emph{confident}, if for \( \ell \) servers, it does not always act defensively.
    An algorithm is \emph{strictly}-\( \ell \)-\emph{confident}, if for \( \ell \) servers, it \emph{never} acts defensively if it can be avoided.
\end{definition}

In the case that multiple servers share the same \( p^{*} \), we only require one of them to move in case of a general request.
Still, all act defensively, as only one must move.

Observe that strictly-\( \ell \)-confident algorithms are also \( \ell \)-confident.
Most algorithms that do not explicitly consider the case of acting defensively are \( \ell \)-confident, e.g., the example algorithm in \Cref{sec:simple-algorithm:unbounded-competitive}.
By chance, they could act defensively for a server, but this is not ensured, for example, simply because any general request is treated the same.
We introduce the notion of strictly-\( \ell \)-confident because it allows us to show a significant drawback when not acting defensively at all for some servers:

\begin{restatable}{thm}{LowerBoundConfidentAlgorithms}\label{th:lower-bounds:confident-algorithms}
    No strictly-\( \ell \)-confident algorithm can achieve a competitive ratio better than \( 2k + \ell - 2 \).
\end{restatable}

\begin{proof}
    Let \alg{} be the online algorithm.
    We consider a uniform metric with locations \( v_{1}, \dots, v_{k+1} \).
    Initially, the algorithm's servers \( a_{1}, \dots, a_{k}  \) as well as \opt{}'s servers \( o_{1}, \dots, o_{k} \) share the same position \( p(a_{i}) = p(o_{i}) = v_{i} \), for all \( i \leq k \).
    As in the proof of \Cref{th:general-lower-bound}, we rename the \( i \)'s such that during the first phase (defined below), \alg{} moves its \( a_{i} \) in order.

    Next, we construct the adversarial sequence.
    First, issue a general request on \( v_{k+1} \).
    Whenever a server \( a_{i} \) is moved by the algorithm, issue a request on \( v_{i} = p^{*}(a_{i}) \) afterwards.
    Either the algorithm acts defensively on \( a_{i} \) and moves it back to \( v_{i} \), or it does not.
    Since \alg{} only sees general requests and is deterministic, by our renaming of the \( i \)'s above, it moves its \( a_{i} \) in order.
    The first phase ends, when \alg{} covers all the positions \( v_{1}, \dots, v_{k-1}, v_{k+1} \).
    From now, we force the algorithm to cover these positions all the time by using sufficiently many general requests on them.
    Let \( \overline{D} \) be the set of servers for which \alg{} did not act defensively and \( D \) be the other servers.
    Let \( L \subseteq \overline{D} \) be the \( \ell \) servers for which the algorithm never acts defensively.

    Observe that for the servers of \( \overline{D} \), \alg{} incurred a cost of \( 1 \) and for any of these servers \( a_{i} \) it holds \( p(a_{i}) \neq p^{*}(a_{i}) \) and \( p(a_{j}) = p^{*}(a_{i}) \) for some \( a_{j} \in \overline{D} \).
    For all other servers, the algorithm had a cost of at least \( 2 \).

    Next, we show that it takes \alg{} additional cost to sort the servers of \( L \) back to their initial position.
    All these servers are in \( \overline{D} \).
    While there is a server in \( a_{i} \in \overline{D} \setminus \{a_{k}\} \) with \( p(a_{i}) \neq p^{*}(a_{i}) \) do the following:
    First, specifically request \( a_{i} \) at \( p^{*}(a_{i}) \).
    Afterward, using general requests, enforce that \alg{} covers \( v_{1}, \dots, v_{k-1}, v_{k+1} \).
    \alg{} has a cost of \( 1 \) for moving \( a_{i} \).
    Additionally, it has a cost of \( 1 \) to move a server to \( a_{i} \)'s previous position.
    Observe that \( a_{i} \)'s previous position \( p \) was either \( v_{k+1} \) or on an initial position of one server of \( \overline{D} \).
    Thus, by moving a server to \( p \), if the moved server is one of the servers for which the algorithm always acts confidently, it cannot return to its initial position because only general requests on \( p \) appear.
    After any iteration, all servers of \( L \) that were not yet specifically requested are at a position different from their initial one.
    This holds because all those servers never act defensively and, at their positions, only general requests appeared.
    Note that \( a_{k} \) cannot be forced back to its initial position.
    Thus, if (a) \( a_{k} \in L \), the process can be repeated \( \ell - 1 \) many times.
    Else, if (b) \( a_{k} \notin L \), the process can be repeated \( \ell \) many times.
    In total, there is a cost of at least \( 2 \, (\ell - 1) \) in case (a), or \( 2 \ell  \) in case (b) for moving servers of \( L \).

    Next, force every other server except \( v_{k} \) back to its initial position using specific requests.
    This implies that every server except \( a_{k} \) has at least a cost of two.

    Summing up, in the case of (a), there are \( k - \ell \) servers having a cost of \( 2 \) while the servers of \( L \) imply a cost of \( \ell + 2(\ell - 1) \).
    In the case of (b), there are \( k - \ell - 1 \) servers (because \( a_{k} \notin L \)) with a cost of \( 2 \) while the servers of \( L \) have a cost of \( \ell + 2 \ell \).
    Thus, in any case, the total cost is at least \( 2k + \ell - 2 \).
    \opt{} solves the entire sequence by moving \( o_{k} \) to \( v_{k+1} \) at a cost of \( 1 \).
    Note that we end up in a situation analogous to the initial one if we rename the locations.
\end{proof}

For \Cref{th:lower-bounds:confident-algorithms}, the input sequence is as depicted in Figure~\ref{figure:lower-bound-sequence}.
Additionally, the servers of \( L \) have further costs.
To see this, consider Figure~\ref{figure:strictly-confident-additional-cost} below.
\\

\begin{minipage}[c]{0.46\textwidth}
    \centering
    \includegraphics[page=5, width=\textwidth, clip=true, trim= 0cm 11cm 19.5cm 0cm]{figures/figures.pdf}
\end{minipage}
\begin{minipage}[c]{0.02\textwidth}
    \hfill
\end{minipage}
\begin{minipage}[c]{0.46\textwidth}
    \centering
    
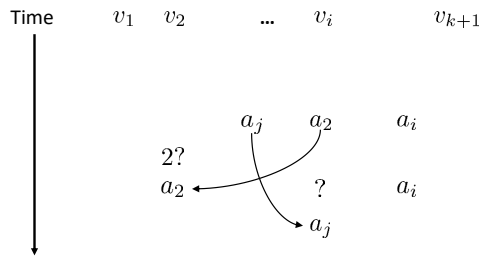
\captionof{figure}{Assume, \( a_{2} \) and \( a_{i} \) are in \( L \).
    Then forcing \( a_{2} \) back to \( v_{2} \) with a specific request incurs cost one and leaves \( v_{i} \) empty.
    A general request on \( v_{i} \) incurs cost one, but since \( a_{i}\in L \), the algorithm cannot move \( a_{i} \) there.
    Therefore, after both requests, \( a_{2} \) is the only server of \( L \) that returned to its initial position.
    As a consequence, each server of \( L \) (except for \( a_{k} \)) can be forced to its initial position with a cost of at least \( 2 \).}
    \label{figure:strictly-confident-additional-cost}
\end{minipage}
\noindent\\

On the other hand, an \( \ell \)-defensive algorithm immediately has a higher upper bound on pure general inputs.

\begin{restatable}{thm}{lowerBoundDefensive}\label{th:lower-bounds:defensive-algorithms}
    No \( \ell \)-defensive algorithm can achieve a competitive ratio better than \( k + \ell - 1 \) on pure general inputs.
\end{restatable}

\begin{proof}
    Let \alg{} be the online algorithm.
    As in the proof of \Cref{th:lower-bounds:confident-algorithms}, consider a uniform metric with positions \( v_{1}, \dots, v_{k+1} \) such that initially, \( p(a_{i}) = p(o_{i}) = v_{i} \) for all servers \( a_{1}, \dots, a_{k}  \) of the algorithm and \( o_{1}, \dots, o_{k} \) of \opt{}.
    As before, assume that the \( i \)'s are renamed such that the algorithm starts to move its \( a_{i} \) in the order of the indices during the phase below.
    Let \( D \) be the set of servers on which \alg{} always acts defensively and \( \overline{D} \) be the other servers.

    The sequence is the first phase of the proof of \Cref{th:lower-bounds:confident-algorithms}:
    First, issue a general request on \( v_{k+1} \).
    Whenever a server \( a_{i} \) is moved by the algorithm, issue a request on \( v_{i} = p^{*}(a_{i}) \) afterwards.
    Either the algorithm acts defensively on \( a_{i} \) and moves it back to \( v_{i} \), or it does not.
    Since \alg{} only sees general requests, is deterministic, and we renamed the \( i \)'s, it moves its \( a_{i} \) in order.
    When \alg{} covers all locations of \( v_{1}, \dots, v_{k-1}, v_{k+1} \), the sequence stops.
    Each server of \( D \) except for possibly \( a_{k} \) was moved by \alg{} twice and each other server at least once.
    Thus, \alg{}'s cost is at least \( 2\, (|D|-1) + |\overline{D}| + 1 \geq k + \ell - 1 \).
    \opt{} solves the entire sequence by moving \( o_{k} \) to \( v_{k+1} \) at a cost of \( 1 \).
    Note that we end up in a situation analogous to the initial one if we rename the locations allowing us to repeat the sequence as many times as wished.
\end{proof}

The proof of \Cref{th:lower-bounds:defensive-algorithms} consists of the first phase of the sequence for \Cref{th:lower-bounds:confident-algorithms}.
For the proof, we use that even without specific requests, the online algorithm moves \( \ell - 1 \) servers back to their initial position and thus twice.
\Cref{th:lower-bounds:confident-algorithms} and \Cref{th:lower-bounds:defensive-algorithms} imply the question: "Can there be an online algorithm that is optimal on pure general inputs and mixed inputs?"
Unfortunately, there cannot be one:

\begin{restatable}{thm}{NoAlgorithmCanBeSuper}\label{th:no-algorithm-can-be-super}
    No deterministic algorithm can achieve a competitive ratio of \( k \) on pure general inputs and \( 2k - 1 \) on mixed inputs.
\end{restatable}

\begin{proof}
    Assume we have such an algorithm.
    As in the proofs of \Cref{th:lower-bounds:confident-algorithms} and \Cref{th:lower-bounds:defensive-algorithms}, consider a uniform metric with positions \( v_{1}, \dots, v_{k+1} \) such that initially, \( p(a_{i}) = p(o_{i}) = v_{i} \) for all servers \( a_{1}, \dots, a_{k}  \) of the algorithm and \( o_{1}, \dots, o_{k} \) of \opt{}.
    Also, as before, assume that the \( i \)'s are renamed such that the algorithm moves its \( a_{i} \) in order in the phase we define below.

    The sequence is as in the proof of \Cref{th:lower-bounds:confident-algorithms}:
    First, issue a general request on \( v_{k+1} \).
    Whenever a server \( a_{i} \) is moved by the algorithm, issue a request on \( v_{i} = p^{*}(a_{i}) \) afterwards.
    \opt{} solves the entire sequence by moving \( o_{k} \) to \( v_{k+1} \) at a cost of \( 1 \).
    Since the algorithm has a competitive ratio of \( k \) on pure general inputs, it acted confidently for all servers.
    Else it would have moved every server at least once and at least one server twice, which implies a cost of at least \( k + 1 \).
    Since the algorithm moved each server in the first phase, it has cost \( k \) for that phase.
    After the first phase, no server is on its initial position and some server \( a_{j} \) with \( j \neq k \) is on \( v_{k+1} \).
    Next, issue a specific request for every server on its position in the optimal solution.
    Then, the algorithm has a cost of at least \( k \) to move every server to its final position.
    The total cost, in this case, is \( 2k > 2k-1 \), and we have a contradiction.
\end{proof}

The proof of \Cref{th:no-algorithm-can-be-super} is based on the input sequences of \Cref{th:lower-bounds:confident-algorithms} and \Cref{th:lower-bounds:defensive-algorithms}.
It uses the fact that any algorithm with a competitive ratio of \( k \) on pure general inputs has to be \( k \)-confident in the first phase.
But then, \( a_{k} \) is not on \( v_{k+1} \).
Thus, even when the algorithm reverts all changes in the second phase and moves \( a_{k} \), it has a cost higher than \( 2k-1 \).

\Cref{th:no-algorithm-can-be-super} still leaves space for algorithms that are close to optimal on all input types.
Due to Theorems \ref{th:lower-bounds:confident-algorithms} and \ref{th:lower-bounds:defensive-algorithms}, such algorithms must act confidently sometimes but not always for a significant number of servers.
Interestingly, most algorithms that treat every general request the same fall into this category, as they might -- by chance -- act defensively.
While our bounds do not yield increased lower bounds for algorithms of the above category, we strongly believe that they suffer similar drawbacks dependent on the number of times they act defensively/confidently.
To see this, observe that the lower bounds are very similar:
The adversarial sequence of \Cref{th:lower-bounds:confident-algorithms} is an extension of the one from \Cref{th:lower-bounds:defensive-algorithms}.
To build the lower bound of \Cref{th:lower-bounds:confident-algorithms} it even suffices if the given algorithm acts confidently for \( \ell \) servers not always but \emph{always in the sequence given by the lower bound}.
As an intuition, to perform well on general inputs, the algorithm should act mostly confidently in the first phase of the lower bound of \Cref{th:lower-bounds:confident-algorithms}.
Then, in the second phase, the algorithm should avoid acting confidently.
Here, we can see that algorithms that treat every general request the same probably perform as badly on mixed inputs as strictly-\( k \)-confident algorithms.
This can also be seen in \Cref{sec:lower-bound-conf}, where we show that for our \( k \)-confident algorithm \textnormal{\textsc{Conf}}, even though it is not strictly-\( k \)-confident, a similar lower bound as \Cref{th:lower-bounds:confident-algorithms} applies.
To conclude, our lower bounds indicate a trade-off between performing well on general/mixed inputs controlled by whether or not and to which degree an online algorithm acts defensively/confidently.

Next, we present two algorithms that incorporate acting defensively/confidently and show bounds on their competitive ratio.
We start by a \( k \)-confident algorithm in \Cref{sec:k-confident-algorithm} followed by a \( k \)-defensive algorithm in \Cref{sec:k-defensive-algorithm}.

\section{A \texorpdfstring{\( k \)}{k}-Confident Algorithm for Uniform Metrics}\label{sec:k-confident-algorithm}

In this section, we present our \( k \)-confident algorithm (see \Cref{def:defensive-confident}) for uniform metrics called \textnormal{\textsc{Conf}}.
We analyze \textnormal{\textsc{Conf}} parameterized in the share of specific requests on the total number of requests.
Note that we only consider requests that require a movement of the algorithm.
All other requests have no relevance to the competitive ratio.
This way, our results (captured in \Cref{th:first-algorithm:competitive-ratio}) not only show how the competitive ratio is bounded for our general model but also for instances of the \( k \)-Server Problem.
For a graphical representation of the theorem, consider \Cref{figure:competitive-ratio-plot}.

\begin{restatable}{thm}{FirstAlgorithmCR}\label{th:first-algorithm:competitive-ratio}
    Let \( g, f \) be the number of general/specific requests that require the algorithm to move.
    Let \( s = \nicefrac{f\,}{\,g + f} \) be the share of specific requests on the total number of requests that require a movement by the algorithm.
    The competitive ratio of \textnormal{\textsc{Conf}} is at most \( \min\{k + \frac{2s}{1-2s}\, k,\, 3k - 2, \, 1 + 2\frac{1-s}{s}\,k \} \) and at most \( \frac{1}{2\, s - 1} \) for \( s > \nicefrac{1}{2} \).
\end{restatable}

\textnormal{\textsc{Conf}} employs ideas of the marking approach \cite[pp.~752-758]{kleinberg_algorithm_2014} and operates in phases.
A phase tries to capture the longest sequence of requests for which the optimal solution does not need to move.
Right after a phase ends, the optimal solution must have a cost of at least \( 1 \).
If the cost of an algorithm in each phase is at most \( c \), it has a competitive ratio of \( c \).
The main difference in our approach, however, is that a server can get unmarked again during a phase.
Also, we do not only differentiate between marked and unmarked servers but distinguish more carefully using set memberships as described below.

\subsection{The Algorithm \textnormal{\textsc{Conf}}}

Next, we describe how \textnormal{\textsc{Conf}} (split into \textsc{Conf-Gen} and \textsc{Conf-Spec}) works.
During the execution, \textnormal{\textsc{Conf}} handles for each phase \( i \) different sets.
We denote that a set belongs to \( i \) by an exponent of \( i \) that is omitted when the phase is clear from the context.

At the beginning of any phase, every server is assigned to a candidate set \( C \) (\textsc{Conf-Gen} Lines 5 -- 6; \textsc{Conf-Spec} Lines 3 -- 4).
During the phase, \textnormal{\textsc{Conf}} handles four sets \( C \), \( G \), \( L \), and \( F \).
We ensure that each server is in exactly one of \( C \), \( G \), or \( F \).
\( L \) stores locations.
More precisely, \( L \) stores the locations where only general requests appeared, and \( G \) stores the servers at such locations.
\( F \), on the other hand, contains all servers that are specifically requested.
Servers of \( F \) can be at the same location while the locations of \( L \) (and thus \( p(G) \)) are distinct and do not overlap with the locations of \( p(F) \).
This distinction is necessary to get a parameterized bound on the competitive ratio in the case where many specific requests appear.
While by definition \( p(G) \subseteq L \), locations in \( L \) can become unoccupied when a server of \( G \) gets specifically requested.

Whenever a general request \( r \) appears, if its location is not in \( p(F) \cup L \), \textnormal{\textsc{Conf}} stores it in \( L \) (\textsc{Conf-Gen} Line 9).
When there is no server of \( G \cup F \) on the requested location already, select a server \( j \in C \) to be the answering server and assign \( j \) to \( G \) (\textsc{Conf-Gen} Line 3, Line 10).
When a specific request \( r \) appears, we observe the following:
The server that is specifically requested must be at the location of \( r \) in the optimal solution.
If we assume that the optimal solution has no cost within the phase, any such server can no longer move after it is specifically requested.
Therefore, \textnormal{\textsc{Conf}} declares server \( j \) as \emph{frozen} and assigns it to \( F \) (\textsc{Conf-Spec} Line 7).
The phase ends when either (1) it can no longer be guaranteed that \( |L| + |F| \leq k \) (\textsc{Conf-Gen} Line 4; \textsc{Conf-Spec} Line 2) or (2) a server \( j \in F \) is specifically requested at a different location (\textsc{Conf-Spec} Line 2).
In case (1), the optimal solution can not cover all locations where requests appeared in the phase with servers, and thus a server must have been moved.
In case (2), the optimal solution must have moved \( j \).

The very first phase is different from all others.
Since we assume that the servers of the online algorithm are at the same locations as the servers of the optimal solution, no movements happen in the first phase.
To reflect this, we set for the first phase \( C^{1}, G^{1}, L^{1} = \emptyset \) and \( F^{1} = K \) (the set of all servers).

\begin{algorithm}[htb]
    \caption*{\textbf{\textnormal{\textsc{Conf-Gen}:} General request \( r \) arrives in phase \( i \)}}
    \begin{algorithmic}[1]
        \If{\( r \notin p(G \cup F) \)}
        \If{\( r \in L \)}
        \State Move some \( j \in C \) to \( r \) and assign it to \( G \)
        \ElsIf{\( |L| + |F| = k \)}
        \State Start the next phase \( i+1 \)
        \State Set \( C^{i+1} \gets K \) and \( G^{i+1}, L^{i+1}, F^{i+1} \gets \emptyset \)
        \State Process \( r \) again for phase \( i+1 \)
        \ElsIf{\( |L| + |F| < k \)}
        \State \( L \gets L \cup \{p(r)\} \)
        \State Move some \( j \in C \) to \( r \) and assign it to \( G \)
        \EndIf
        \EndIf
    \end{algorithmic}
\end{algorithm}

\begin{algorithm}[htb]
    \caption*{\textbf{\textnormal{\textsc{Conf-Spec}:} Specific request \( r \) for server \( j \) arrives in phase \( i \)}}
    \begin{algorithmic}[1]
        \If{\( r \neq p(j) \)}
        \If{\( j \in F \) \textbf{or} \( |L| + |F| = k \)}
        \State Start the next phase \( i + 1 \)
        \State Set \( C^{i+1} \gets K  \) and \( G^{i+1}, L^{i+1}, F^{i+1} \gets \emptyset \)
        \State Process \( r \) again for phase \( i+1 \)
        \ElsIf{\( j \notin F \) and \( |L| + |F| < k \)}
        \State Move \( j \) to \( r \) and assign it to \( F \)
        \If{There is a \( s \neq j \), \( s \notin F \) on \( r \)}
        \State Remove \( p(s) \) from \( L \)
        \State Assign \( s \) to \( C \)
        \EndIf
        \EndIf
        \Else
        \State Assign \( j \) to \( F \)
        \EndIf
    \end{algorithmic}
\end{algorithm}

We remark that the statement to move any server of \( C \) to serve a general request is ambiguous, and any order on the servers of \( C \) will do.
For the sake of precision, assume that the servers are selected using the FIFO (first-in-first-out) rule.

Due to specific requests, \textnormal{\textsc{Conf}} incorporates behaviors that are fundamentally different from the classical \( k \)-Server Problem:
Observe that a specific request removes a location \( x \in L \) when a server becomes frozen there.
When this happens, a server \( j \in G \) can even be assigned to \( C \) again (\textsc{Conf-Spec} Lines 8 -- 10).
From a perspective of a marking algorithm, this means \( j \) becomes unmarked again.
Intuitively, by the specific request, \textnormal{\textsc{Conf}} detects that \( j \) was the wrong server to answer the previous general request on \( x \).
Moreover, a specific request for \( j \) can yield that \( j \)'s previous location of \( G \) becomes \emph{unoccupied}.
To still keep track of it, \textnormal{\textsc{Conf}} stores it in \( L \).
For such a location of \( L \) where no server of \( G \) is, it may be necessary to move another server on it due to a later general request.
One could ensure that all locations of \( L \) are covered the entire time by servers in \( G \), but there is no advantage in that behavior.

\subsection{The Analysis}

For the analysis, we begin by formally showing that the cost of the optimal solution \opt{} is bounded by the number of phases in \Cref{le:first-algorithm:opt-cost:1-for-each-phase}.
This is used in combination with a bound on the cost of \textnormal{\textsc{Conf}} in each phase (\Cref{le:first-algorithm:algs-cost}) to show a worst-case bound (\Cref{le:first-algorithm:worst-case-comp-ratio}) and an adaptive bound (\Cref{le:first-algorithm:adaptive-comp-ratio}) for the competitive ratio.

We denote by \( \widehat{U^{i}} \) the content of the set \( U^{i} \) right after the end of phase \( i \).
For a phase, let \( t_{\text{start}} \) be the time step of the first request and \( t_{\text{end}} \) be the time step of the last request.
\Cref{le:first-algorithm:opt-cost:locations-disjoint} is ensured by the algorithm.

\begin{restatable}{lem}{FirstAlgorithmLocationsDisjoint}\label{le:first-algorithm:opt-cost:locations-disjoint}
    At any point in time, \( L \) and \( p(F) \) are disjoint.
\end{restatable}

\begin{proof}
    Assume there is a location \( \ell \in L \) where some server \( s \in F \) is.
    If \( s \) was first on \( \ell \), no server of \( C \) would be moved to \( \ell \) and \( \ell \) would not be in \( L \), because \( s \) is able to answer general requests.
    If \( \ell \) became part of \( L \) first, \( s \) moved to \( \ell \) due to some specific request.
    Then the algorithm ensures that \( \ell \) is no longer part of \( L \).
\end{proof}

Now we can bound the cost of an optimal solution.

\begin{restatable}{lem}{FirstAlgorithmOptCostOne}\label{le:first-algorithm:opt-cost:1-for-each-phase}
    Consider any phase but the last and the first request \( r \) right after the phase ends.
    \opt{} has a cost of at least \( 1 \) during the time right after \( t_{\text{start}} \) until right after \( t_{\text{end}} + 1 \).
\end{restatable}

\begin{proof}
    Consider any phase that ends.
    At the end of the phase it holds either that (i) \( |\widehat{L}| + |\widehat{F}| = k \) or (ii) a server \( j \in F \) is specifically requested at a location different to \( p^{*}(j) \).
    Assume \opt{} has had no movement.
    Then, during the time interval, \opt{} has its servers at least at the locations \( L \cup p(\widehat{F}) \cup \{r\} \) since at each of these locations a request appeared (right after \( t_{\text{start}} \), \opt{} has a server on the point of the first request).
    Also, \opt{} must have the servers of \( F \) at the identical locations as \textnormal{\textsc{Conf}}.

    In the case of (i), \opt{} covers due to \Cref{le:first-algorithm:opt-cost:locations-disjoint} \( |\widehat{L}| + 1 > k - |\widehat{F}| \) distinct locations using \( k - |\widehat{F}| \) servers which cannot be.
    In case of (ii), \( j \) is specifically requested at two different locations meaning that \opt{} must have placed \( j \) at two different locations.
    In any case, there is a contradiction and, thus, \opt{} has cost at least \( 1 \).
\end{proof}

Next, we show how the cost of the algorithm for a phase is bounded.
For technical reasons, we restate the cost of the algorithm in a phase as follows.

\begin{restatable}{obs}{FirstAlgorithmsCost}\label{le:first-algorithm:algs-cost}
    Consider any phase \( i > 1 \).
    Let \( g^{i} \) be the number of general requests during the phase that require a movement of the algorithm.
    Let \( f^{i} \) be the number of specific requests during the phase that require a movement of the algorithm.
    The cost of \textnormal{\textsc{Conf}} in the phase is at most \( g^{i} + f^{i} \).
\end{restatable}

In the following lemma, we show the worst-case upper bound for our algorithm.

\begin{restatable}{lem}{FirstAlgorithmWorstCaseCR}\label{le:first-algorithm:worst-case-comp-ratio}
    The competitive ratio of \textnormal{\textsc{Conf}} is at most \( 3k - 2 \).
\end{restatable}

\begin{proof}
    Assume there are \( p \) phases.
    Due to \Cref{le:first-algorithm:opt-cost:1-for-each-phase}, \opt{}'s cost is at least \( p-1 \).
    Based on \Cref{le:first-algorithm:algs-cost}, we know that for any phase but the first, \textnormal{\textsc{Conf}}'s cost is at most \( g^{i} + f^{i} \).
    Observe that it also holds (1) \( g^{i} \leq |\widehat{G}| + |\widehat{F}| + f^{i} \), (2) \( f^{i} \leq |\widehat{F}| \), and (3) \( |\widehat{G}| + |\widehat{F}| \leq k \).
    (1) holds because general requests requiring a movement can only have appeared at locations covered by a server in the end and at unoccupied locations.
    The number of former locations are upper bounded by \( |\widehat{G}| + |\widehat{F}| \).
    A location can only become unoccupied when a server moves away from it and joins \( F \), which implies that there are at most \( f^{i} \) many.
    (2) holds by definition, and (3) is ensured by the algorithm.

    We consider two cases: (a) \( |\widehat{F}| < k \) and (b) \( |\widehat{F}| = k \).
    Consider the case of (a).
    Then in the phase, the cost of the algorithm is at most \( g^{i} + f^{i} \leq |\widehat{G}| + 3 \, |\widehat{F}| \leq k + 2\,(k-1) \) (using (1), (2) and the premise).
    Consider the case of (b).
    In this case, \( |\widehat{G}| = 0 \).
    Consider the last specific request for server \( j \).
    Either \( j \) is already at the requests' position, or it was not used before, i.e., \( j \in C \).
    If \( j \) is already at the request's position, \( f^{i} < |\widehat{F}| \) and thus, \( g^{i} + f^{i} \leq 3\,|\widehat{F}| - 2 \).
    Else, \( j \) was also in \( C \) when the penultimate specific request appeared, implying that due to that request, no position of \( G \) became unoccupied.
    Therefore, \( g^{i} \leq 2 |\widehat{F}| - 2 \).
    In total, the cost for the phase is thus in any case \( g^{i} + f^{i} \leq 3\, |\widehat{F}| - 2 \leq 3\, k - 2 \) and the cost of the algorithm over all phases is thus at most \( (p-1) (3k - 2) \).
    Since \opt{} has a cost of at least \( p-1 \), the lemma follows.
\end{proof}

Next, we show three bounds that parameterize the competitive ratio of \textnormal{\textsc{Conf}} in the structure of the input sequence.

\begin{restatable}{lem}{FirstAlgorithmAdaptiveCR}\label{le:first-algorithm:adaptive-comp-ratio}
    Let \( g, f \) be the number of general/specific requests that require the algorithm to move.
    Let \( s = \nicefrac{f\,}{\,g+f} \) be the share of specific requests on the total number of requests that require a movement by the algorithm.
    The competitive ratio of \textnormal{\textsc{Conf}} on uniform metrics is at most \( \min\{k + \frac{2s}{1-2s}\, k,\, 1 + 2\frac{1-s}{s}\,k \} \).
    For \( s > \frac{1}{2} \), the competitive ratio is at most \( \frac{1}{2\,s - 1} \).
\end{restatable}

\begin{proof}
    For the proof, for any phase, we consider two disjoint sets of servers composing \( \widehat{F} \).
    Servers in \( \widehat{F_{1}} \) are those which were specifically requested at the exact location as they were the last time before the phase, while servers in \( \widehat{F_{2}} \) were specifically requested at a location different from the previous one.

    Assume there are \( p \) phases.
    Denote the cost of the optimal offline solution by \( c(\opt{}) \).
    First, due to \Cref{le:first-algorithm:opt-cost:1-for-each-phase}, \opt{}'s cost is at least \( p-1 \).
    Every time a server gets specifically requested at a location different than the one where it was specifically requested the last time, \opt{} has to move the server.
    Therefore, secondly, \opt{} has a cost of at least \( \sum_{i} |\widehat{F^{i}_{2}}| \).

    Next, consider the cost of \textnormal{\textsc{Conf}} denoted by \( c(\textnormal{\textsc{Conf}}) \).
    We start with some basics.
    As before in the proof of \Cref*{le:first-algorithm:worst-case-comp-ratio}, for any phase but the first, \textnormal{\textsc{Conf}}'s cost is at most \( g^{i} + f^{i} \) and it holds (1) \( g^{i} \leq |\widehat{G}| + |\widehat{F}| + f^{i} \), (2) \( f^{i} \leq |\widehat{F}| \) and (3) \( |\widehat{G}| + |\widehat{F}| \leq k \).

    Observe that the first phase costs the algorithm nothing, i.e., \( f^{1} = g^{1} = 0 \).
    We can derive (4) \( \sum_{i=2}^{p} f^{i} \leq \frac{s}{1-2s} \sum_{i=2}^{p} k \) as follows:
    \( s = \nicefrac{\sum_{i} f^{i}\,}{\,\sum_{i} (g^{i} + f^{i})} \) \( \Leftrightarrow s \, \sum_{i} (g^{i} + f^{i}) = \sum_{i} f^{i} \) \( \Leftrightarrow s \, \sum_{i} (k + 2\,f^{i})  \geq \sum_{i} f^{i} \Leftrightarrow \sum_{i=2}^{p} f^{i} \leq \frac{s}{1-2s} \sum_{i=2}^{p} k \).
    Also, we can derive (5) \( \sum_{i=2}^{p} g^{i} \leq \frac{1-s}{s} \sum_{i=2}^{p} k \) by \( s = \nicefrac{\sum_{i} f^{i}\,}{\,\sum_{i} (g^{i} + f^{i})} \Leftrightarrow s \, \sum_{i} g^{i} = (1-s) \sum_{i} f^{i} \Leftrightarrow \sum_{i=2}^{p} g^{i} = \frac{1-s}{s} \sum_{i=2}^{p} f^{i} \) \( \leq \frac{1-s}{s} \sum_{i=2}^{p} k \).

    We begin by using (1), (2), (3), and (4).
    Summed up over all phases, we end up at:
    \begin{align*}
        c(\textnormal{\textsc{Conf}}) & \leq \sum_{i=2}^{p} (g^{i} + f^{i})
        \leq \sum_{i=2}^{p} (|\widehat{G^{i}}| + |\widehat{F^{i}}| + 2 f^{i})
        \leq \sum_{i=2}^{p} \left(k + 2\,f^{i} \right)                       \\
                         & \leq (p-1) \left(k + \frac{2s}{1-2s} \, k \right)
        \leq \left(k + \frac{2s}{1-2s} \, k \right) \, c(\opt{}).
    \end{align*}

    Next, we turn to the second bound.
    Let \( f^{i}_{1} \) be the number of movements due to servers in \( \widehat{F^{i}_{1}} \) and \( f^{i}_{2} \) be the number of movements due to servers in \( \widehat{F^{i}_{2}} \).
    Consider the servers of \( \bigcup_{i} \widehat{F^{i}_{1}} \).
    For any such server \( j \) with respect to phase \( i \), it holds: If our algorithm has a cost for \( j \) when \( j \) joins \( \widehat{F^{i}_{1}} \), then \( j \) was moved by a general request since the time it was lastly specifically requested.
    This implies that (6) \( \sum_{i} f^{i}_{1} \leq \sum_{i} g^{i} \).
    Using (1), (2), (3), (5), and (6) yields:
    \begin{align*}
        c(\textnormal{\textsc{Conf}}) & \leq \sum_{i=2}^{p} (g^{i} + f^{i})
        \leq \sum_{i=2}^{p} (g^{i} + f^{i}_{1} + f^{i}_{2})
        \leq \sum_{i=2}^{p} (2\,g^{i} + f^{i}_{2})                                                      \\
                         & \leq \sum_{i=2}^{p} \left(2\,\frac{1-s}{s}\,k + |\widehat{F^{i}_{2}}|\right)
        \leq (p-1)\,2\,\frac{1-s}{s}k + \sum_{i} |\widehat{F^{i}_{2}}|                                  \\
                         & \leq \left(1 + 2\,\frac{1-s}{s}k\right) \, c(\opt{}).
    \end{align*}

    Consider now the third bound.
    By the derivation of (5) we know (7) \( \sum_{i=2}^{p} g^{i} \leq \frac{1-s}{s} \sum_{i=2}^{p} f^{i} \).
    Based on (6) and (7) and because \( \sum_{i} f_{2}^{i} \geq 0 \) and \( s > \nicefrac{1}{2} \), we can derive (8) \( \sum_{i=2}^{p} f^{i}_{1} \leq \frac{1-s}{2 \, s - 1} \, \sum_{i=2}^{p} f^{i}_{2} \) by \( \sum_{i=2}^{p} f^{i}_{1} \leq \sum_{i=2}^{p} g^{i} \leq \frac{1-s}{s} \sum_{i=2}^{p} f^{i} = \frac{1-s}{s} \sum_{i=2}^{p} (f_{1}^{i} + f_{2}^{i}) \Leftrightarrow (\frac{s}{s} - \frac{1-s}{s}) \sum_{i=2}^{p} f_{1}^{i} \leq \frac{1-s}{s} \sum_{i=2}^{p} f_{2}^{i} \Leftrightarrow \sum_{i=2}^{p} f^{i}_{1} \leq \frac{1-s}{2 \, s - 1} \, \sum_{i=2}^{p} f^{i}_{2} \).
    Using (7) and (8) then yields:
    \begin{align*}
        c(\textnormal{\textsc{Conf}}) & \leq \sum_{i=2}^{p} (g^{i} + f^{i})
        \leq \left(1 + \frac{1-s}{s}\right) \, \sum_{i=2}^{p} f^{i}
        \\
        & = \left(1 + \frac{1-s}{s}\right) \, \sum_{i=2}^{p} f^{i}_{1} + \left(1 + \frac{1-s}{s}\right) \, \sum_{i=2}^{p} f^{i}_{2}
        \\
        & \leq \left(1 + \frac{1-s}{s}\right) \cdot \frac{1-s}{2 \, s - 1} \, \sum_{i=2}^{p} f^{i}_{2} + \left(1 + \frac{1-s}{s}\right) \, \sum_{i=2}^{p} f^{i}_{2}
        \\
        & \leq \frac{1}{2\,s - 1} \, \sum_{i=2}^{p} f^{i}_{2}
        \leq \frac{1}{2\,s - 1} \, \sum_{i=2}^{p} |\widehat{F_{2}^{i}}|
        \leq \frac{1}{2\,s - 1} \, c(\opt{}).
    \end{align*}
\end{proof}

\Cref{th:first-algorithm:competitive-ratio} now follows from \Cref{le:first-algorithm:worst-case-comp-ratio} and \Cref{le:first-algorithm:adaptive-comp-ratio}.

\subsection{\textnormal{\textsc{Conf}} has a Worst-Case Competitive Ratio of at least \texorpdfstring{\( 3k-2 \)}{3k-2}}\label{sec:lower-bound-conf}

In this section, we show that there is a mixed input for \textnormal{\textsc{Conf}} such that the algorithm's competitive ratio is at least \( 3k - 2 \).
I.e., even though \textnormal{\textsc{Conf}} is not strictly-\( k \)-confident but only \( k \)-confident, the same lower bound as the one of \Cref{th:lower-bounds:confident-algorithms} applies.

\begin{restatable}{thm}{LowerBoundConfAlgorithm}
    The worst-case competitive ratio of \textnormal{\textsc{Conf}} is at least \( 3k-2 \).
\end{restatable}

\begin{proof}
    We assume \textnormal{\textsc{Conf}} selects servers of \( C \) by the FIFO rule.
    As a remark, the bound below can be adapted for other orders for the selection.

    Our lower bound is constructed as the lower bound of \Cref{th:lower-bounds:confident-algorithms}:
    We consider a uniform metric with locations \( v_{1}, \dots, v_{k+1} \).
    Initially, the algorithm's servers \( a_{1}, \dots, a_{k}  \) as well as \opt{}'s servers \( o_{1}, \dots, o_{k} \) share the same position \( p(a_{i}) = p(o_{i}) = v_{i} \), for all \( i \leq k \).
    As in the proof of \Cref{th:general-lower-bound}, we rename the \( i \)'s such that during the first phase (defined below), the algorithm moves its \( a_{i} \) in order.
    First, issue a general request on \( v_{k+1} \).
    \opt{} solves the entire sequence by moving \( o_{k} \) to \( v_{k+1} \) at a cost of \( 1 \).
    Whenever a server \( a_{i} \) is moved by the algorithm, issue a request on \( v_{i} = p^{*}(a_{i}) \) afterwards, except for \( a_{k} \).
    After the first phase, \textnormal{\textsc{Conf}} covers the locations \( 1, \dots, k-1 \) and \( k+1 \) in the following way:
    \( a_{1} \) is on \( v_{k+1} \) and each \( a_{i} \) for \( i > 1 \) is on \( v_{i-1} \).
    \textnormal{\textsc{Conf}} has moved every server once, i.e., all servers are in \( G \).
    Now, for each \( i < k \), issue a specific request on \( v_{i} \) for server \( i \) and afterwards, a general request on \( a_{i} \)'s previous position.
    For each such two requests, \textnormal{\textsc{Conf}} moves server \( a_{i} \) to \( v_{i} \), the server \( a_{i+1} \) joins \( C \) and immediately joins \( G \) on \( v_{k+1} \).
    In total, we can do \( k-1 \) such pairs of requests until all servers except for \( a_{k} \) are on their initial position and \( a_{k} \) is on \( v_{k+1} \).
    At this point, \textnormal{\textsc{Conf}}'s configuration matches the optimal one.
    The cost of \textnormal{\textsc{Conf}} is then \( k + 2 \, (k-1) = 3k - 2 \) while \opt{} has had a cost of \( 1 \).
\end{proof}

\section{A \texorpdfstring{\( 2k + 14 \)}{2k+ 14} competitive Algorithm for Uniform Metrics}\label{sec:k-defensive-algorithm}

In the following, we present \textnormal{\textsc{Def}}, a \( k \)-defensive algorithm (see \Cref{def:defensive-confident}) achieving a worst-case competitive ratio of \( 2k + 14 \) (\Cref{th:defensive-algorithm:competitive-ratio}) on uniform metrics.
\textnormal{\textsc{Def}} comes close to the general lower bound of \( 2k - 1 \) (see \Cref{sec:lower-bound}).
Similar to \textnormal{\textsc{Conf}}, the algorithm is loosely inspired by the marking approach.
\textnormal{\textsc{Def}} can be seen as an extended version of \textnormal{\textsc{Conf}}, where for each server, the algorithm acts defensively.

\begin{restatable}{thm}{SecondAlgorithmCR}\label{th:defensive-algorithm:competitive-ratio}
    The competitive ratio of \textnormal{\textsc{Def}} on uniform metrics is at most \( 2 k + 14 \).
\end{restatable}

\subsection{The Algorithm \textnormal{\textsc{Def}}}

As \textnormal{\textsc{Conf}} (\Cref{sec:k-confident-algorithm}) does, \textnormal{\textsc{Def}} works in phases and is split into \textsc{Def-Gen} and \textsc{Def-Spec}.
As before, \textnormal{\textsc{Def}} manages several sets for each phase \( i \).
We denote this by an exponent of \( i \) that is omitted if the phase is clear from the context.

At the beginning of a phase, all servers are in a candidate set \( C \) (\textsc{Def-Gen}, \textsc{Def-Spec} Lines 3 -- 4).
Similar to \textnormal{\textsc{Conf}} we have the sets \( G \) and \( F \).
As before, \( F \) is the set of servers for which specific requests appeared so far during phase \( i \) (as a result, these servers become \emph{frozen}).
\( G \) is defined as the set of servers at locations where only general requests appeared of which no server acted defensively.
In contrast to the definition for \textnormal{\textsc{Conf}}, we do not allow locations where only general requests appeared to become unoccupied.
Thus, we do not need \( L \).
This change does not influence the worst-case bound and improves the readability.
The worst-case bound is unaffected because, at every unoccupied location, a general request increases only the cost of the online algorithm and hence, happens in the worst case.
To ensure that no unoccupied location appears, \textnormal{\textsc{Def}} simulates a general request whenever a server not in \( C \) moves away from a location (\textsc{Def-Gen}, \textsc{Def-Spec} Lines 13 -- 14).
In addition to the above sets, we have a set \( D \) of servers that acted defensively during the current phase.
\( D \) and \( G \) are disjoint, and together, they contain all servers that are at locations where only general requests appeared.
Intuitively, \textnormal{\textsc{Conf}} treats all servers moving due to general requests the same, while \textnormal{\textsc{Def}} acts defensively whenever possible.

With the same reasoning as in the description of \textnormal{\textsc{Conf}}, whenever a specific request for server \( j \) appears, \( j \) is never moved for the rest of the phase and thus joins \( F \) (\textsc{Def-Spec} Line 7).
When a general request \( r \) appears, \textnormal{\textsc{Def}} first determines if there is a server \( j \in C \cup G \) such that \( p^{*}(j) = r \) (\textsc{Def-Gen} Lines 10 -- 12).
Note how by definition, no server of \( D \cup F \) can act defensively for \( r \) on a location different from its current one.
If so, the algorithm acts defensively by moving \( j \) to \( r \), and assigns \( j \) to \( D \).
As a tie break, when there are multiple such servers, \textnormal{\textsc{Def}} picks the one which was specifically requested the latest.
If no such server exists, \textnormal{\textsc{Def}} moves a server of the candidate set \( C \) to \( r \) and assign it to \( G \) (\textsc{Def-Gen} Lines 7 -- 9).
The respective server is chosen by a scheme that prioritizes servers that did not yet move in the current phase (\textsc{Def-Select} Lines 1 -- 6) and those that would not act defensively if there is already some other server that acts defensively for the same location (\textsc{Def-Select} Lines 2 -- 4).
For details, see the algorithm \textsc{Def-Select} below.
Note here, how \( C \) and \( G \) are split into \( C_{1} \) and \( C_{2} \), and \( G_{1} \) and \( G_{2} \) to keep track of servers that acted defensively.
\( C_{1} \) and \( G_{1} \) contain servers that never joined \( D \) during the current phase, while \( C_{2} \) and \( G_{2} \) contain those servers that were in \( D \) earlier.
Note how \textsc{Def-Select} is ambiguous for the real choice of a server of \( C_{1} \) or \( C_{2} \).
As in \textnormal{\textsc{Conf}}, any ordering on the servers will do, and we assume for this paper that the FIFO rule is used.

A phase ends when either a server of \( F \) is specifically requested at a different location (\textsc{Def-Spec} Line 2) or when \( |G| + |D| + |F| \leq k \) would not hold anymore (\textsc{Def-Gen}, \textsc{Def-Spec} Line 2).
In the former case, the optimal solution needs to move the respective server for a cost of at least \( 1 \).
In the latter case, more than \( k - |F| \) locations would need to be covered by \( k - |F| \) servers which implies that the optimal solution has cost at least \( 1 \).

As before in \textnormal{\textsc{Conf}}, we assume that initially, the servers are at the exact locations as in the optimal solution.
Hence, \( C^{1}, G^{1}, D^{1} = \emptyset \) and \( F^{1} = K \) (the set of all servers).

\begin{algorithm}[H]
    \caption*{\textbf{\textnormal{\textsc{Def-Select}:} Server for request \( r \)}}
    \begin{algorithmic}[1]
        \If{\( C_{1} \) is not empty}
        \If{There is \( j \in C_{1} \) such that \( j \) would not be selected \\ \hspace*{1.1cm} to act defensively for \( p^{*}(j) \)}
        \State \textbf{Return} \( j \)
        \Else
        \State \textbf{Return} Any server of \( C_{1} \)
        \EndIf
        \ElsIf{\( C_{1} \) is empty}
        \State \textbf{Return} Any server of \( C_{2} \)
        \EndIf
    \end{algorithmic}
\end{algorithm}

\begin{algorithm}[H]
    \caption*{\textbf{\textnormal{\textsc{Def-Gen}:} General request \( r \) arrives in a phase \( i \)}}
    \label{algorithm:fancy-algorithm-general-v2}
    \begin{algorithmic}[1]
        \If{\( r \notin p(G \cup D \cup F)  \)}
        \If{\( |G| + |D| + |F| = k \)}
        \State Start the next phase \( i+1 \)
        \State Set \( C^{i+1} \gets K \) and \( G^{i+1}, D^{i+1}, F^{i+1} \gets \emptyset \)
        \State Process \( r \) again for phase \( i+1 \)
        \ElsIf{\( |G| + |D| + |F| < k \)}
        \If{\( r \notin p^{*}(C) \cup p^{*}(G) \)}\label{algorithm:fancy-algorithm-general-v2:line:answer-new-request}
        \State Pick server \( j \in C \) given by \textsc{Select}
        \State Move \( j \) to \( r \) and assign it to \( G \)
        \Statex \hspace*{1.65cm} (\( G_{1} \) if it was in \( C_{1} \), \( G_{2} \) else)
        \ElsIf{\( r \in p^{*}(C) \cup p^{*}(G) \)}\label{algorithm:fancy-algorithm-general-v2:line:move-server-back}
        \State Let \( j \notin F \) be the server with \( p^{*}(j) = r \)
        \Statex \hspace*{1.65cm} If there are multiple, select the one that
        \Statex \hspace*{1.65cm} was specifically requested the latest.
        \State Move \( j \) to \( r \) (\( p^{*}(j) \)) and assign it to \( D \)
        \If{\( j \) was in \( G \)}
        \State Simulate a general request on \( j \)'s
        \Statex \hspace*{2.15cm} previous position
        \EndIf
        \EndIf
        \EndIf
        \EndIf
    \end{algorithmic}
\end{algorithm}

\begin{algorithm}[H]
    \caption*{\textbf{\textnormal{\textsc{Def-Spec}:} Specific request \( r \) for server \( j \) arrives in phase \( i \)}}
    \label{algorithm:fancy-algorithm-specific-v2}
    \begin{algorithmic}[1]
        \If{\( r \neq p(j) \)}
        \If{\( j \in F \) \textbf{or} \( |G| + |D| + |F| = k \)}
        \State Start the next phase \( i + 1 \)
        \State Set \( C^{i+1} \gets K \) and \( G^{i+1}, D^{i+1}, F^{i+1} \gets \emptyset \)
        \State Process \( r \) again for phase \( i+1 \)
        \ElsIf{\( j \notin F \) and \( |G| + |D| + |F| < k \)}
        \State Move \( j \) to \( r \) and assign it to \( F \)
        \If{There is a \( s \neq j \), \( s \notin F \) on \( r \)}
        \If{\( s \) was in \( C_{1} \cup G_{1} \)}
        \State Assign \( s \) to \( C_{1} \)
        \Else
        \State Assign \( s \) to \( C_{2} \)
        \EndIf
        \EndIf
        \If{\( j \) was in \( G \cup D \)}\label{algorithm:fancy-algorithm-specific-v2:line:replace-server-j}
        \State Simulate a general request on \( j \)'s
        \Statex \hspace*{1.65cm} previous position
        \EndIf
        \EndIf
        \Else
        \State Assign \( j \) to \( F \) \label{algorithm:fancy-algorithm-specific-v2:line:j-freezes}
        \EndIf
    \end{algorithmic}
\end{algorithm}

\subsection{The Analysis}

Next, we show that \textnormal{\textsc{Def}} has a competitive ratio of \( 2 k + 14 \).
The starting approach is the same as in the analysis of \textnormal{\textsc{Conf}}, i.e., we bound the cost of \textnormal{\textsc{Def}} in each phase and use that \opt{} has cost \( 1 \) in each phase.
However, the cost of \textnormal{\textsc{Def}} might be higher than \( 2k + 14 \) in each phase.
To reason about these higher costs, we first analyze in detail which costs \textnormal{\textsc{Def}} produces.
Afterward, we simplify the bound step-by-step using insights into the behavior of \textnormal{\textsc{Def}}.
Then, we show how to charge the simplified cost of \textnormal{\textsc{Def}} in a phase to movements of \opt{}.
For this step, we identify further movements of \opt{} that must happen to answer specific requests.

Before we start, observe that \textnormal{\textsc{Def}} is \( k \)-defensive with respect to \Cref{def:defensive-confident}.
This is ensured by Lines 10-14 for serving a general request.
If there is a server that can act defensively for \( r \), then \( r \in p^{*}(C) \cup p^{*}(G) \) holds.
In the respective lines, \textnormal{\textsc{Def}} selects a server that acts defensively for \( r \), and we assign it to \( D \).

\paragraph{On the Cost of \opt{} in a Phase.}
As before we denote by \( \widehat{U^{i}} \) the content of the set \( U^{i} \) right after the end of phase \( i \).
For a phase, let \( t_{\text{start}} \) be the time step of the first request and \( t_{\text{end}} \) be the time step of the last request.
\Cref{le:opt-cost:locations-disjoint} is an adaption of \Cref{le:first-algorithm:opt-cost:locations-disjoint}.

\begin{restatable}{lem}{SecondAlgorithmLocationsDisjoint}\label{le:opt-cost:locations-disjoint}
    At any point in time \( p(G) \), \( p(D) \), and \( p(F) \) are disjoint.
\end{restatable}

\begin{proof}
    Assume there is one location \( \ell \) with a \( j \in G \) and some \( s \in D \).
    Both servers joined their sets due to a general request on \( \ell \).
    No matter which server joined its set first, the later one would not join its set on \( \ell \) because there was already a server on that location.

    Assume there is a location \( \ell \) with some server \( j \in G \cup D \) and some server \( s \in F \).
    If \( s \) was first on \( \ell \), \( j \) would not have joined its set there, because \( s \) is able to answer general requests.
    If \( j \) was first on \( \ell \), \( s \) moved to \( \ell \) due to some specific request.
    Then the algorithm ensures that \( j \) leaves \( G \cup D \) and joins \( C \).
\end{proof}

\Cref{le:opt-cost:1-for-each-phase} is an adaption of \Cref{le:first-algorithm:opt-cost:1-for-each-phase} taking \( D \) into account.

\begin{restatable}{lem}{SecondAlgorithmOptCostOne}\label{le:opt-cost:1-for-each-phase}
    Consider any phase but the last and the first request \( r \) right after the phase ends.
    \opt{} has a cost of at least \( 1 \) during the time interval right after \( t_{\text{start}} \) until right after \( t_{\text{end}} + 1 \).
\end{restatable}

\begin{proof}
    Consider any phase that ends.
    At the end of the phase it holds either that (i) \( |\widehat{G}| + |\widehat{D}| + |\widehat{F}| = k \) or (ii) a server \( j \in F \) is specifically requested at a location different to \( p^{*}(j) \).
    Assume \opt{} has had no movement.
    Then, during the time interval, \opt{} has its servers at least at the locations \( p(\widehat{G} \cup \widehat{D} \cup \widehat{F}) \cup \{r\} \) since at each of these locations a request appeared (at \( t_{\text{start}} + 1 \), \opt{} has a server on the point of the first request).
    Also, \opt{} must have each server of \( F \) at the identical location as \textnormal{\textsc{Def}}.

    In the case of (i), \opt{} covers due to \Cref{le:opt-cost:locations-disjoint} \( |\widehat{G}| + |\widehat{D}| + 1 > k - |\widehat{F}| \) distinct locations using \( k - |\widehat{F}| \) servers which cannot be.
    In case of (ii), \( j \) is specifically requested at two different locations meaning that \opt{} must have placed \( j \) at two different locations.
    In any case, there is a contradiction and, thus, \opt{} has cost at least \( 1 \).
\end{proof}

Next, we show that depending on the configuration of \opt{} after a phase, the optimal cost can even be higher during the phase.

\begin{restatable}{lem}{SecondAlgorithmOptCostNotAllCovered}\label{le:opt-cost:high-when-not-all-covered}
    For a phase, let \( p_{1} \) be the number of locations of \( p(\widehat{G} \cup \widehat{D}) \) where no server of \opt{} is located and let \( p_{2} \) be the number of servers in \( \widehat{F} \) that \opt{} has not located where they were specifically requested.
    Let \( p:= p_{1} + p_{2} \), then \opt{} has cost at least \( p \) during the time interval right after \( t_{\text{start}} \) until right after \( t_{\text{end}} \).
\end{restatable}

\begin{proof}
    Right after \( t_{end} \), it holds that there are \( p_{1} \) locations where requests appeared during the phase and where \opt{} has no server.
    Consider any such location.
    Since a request appeared there, \opt{} must have had a server on it during the phase.
    Since there is no server on it after the phase, \opt{} moved its server for a cost of \( 1 \).
    Consider any server contributing to \( p_{2} \).
    During the phase, \opt{} must have had the respective server on the position where it is specifically requested.
    Since this is no longer the case after the phase, \opt{} moved it for a cost of \( 1 \).
\end{proof}

Intuitively, \opt{} must have a server at all locations that appear during the phase, and hence, not covering all implies an equal movement cost.
Observe that for any phase but the last the cost is \( \max\{1, p\} \) while in the last phase, the cost is \( p \), as it ends by definition at \( t_{end} \).

\paragraph{On the Cost of \textnormal{\textsc{Def}} in a Phase.}
Next, we analyze the cost of \textnormal{\textsc{Def}} for any but the first phase.
In the first phase, all servers are frozen; therefore, \textnormal{\textsc{Def}} has cost zero.
Before we start, we state \Cref{le:second-algorithm:on-g-no-defensive}.
It holds because our algorithm acts defensively for all servers.

\begin{restatable}{lem}{SecondAlgorithmOnGNoDefensive}\label{le:second-algorithm:on-g-no-defensive}
    If at any time on a location \( \ell \), a server \( j \) joins \( G \), there is no server \( s \) with \( p^{*}(s) = \ell \) until the next time a server is specifically requested on \( \ell \).
\end{restatable}

\begin{proof}
    Consider such a location and assume, there is a \( s \) with \( p^{*}(s) = \ell \).
    Then, since \( s \) was not specifically requested on \( \ell \) since \( j \) joined \( G \) on \( \ell \), \( p^{*}(s) = \ell \) at the time when \( j \) joined \( G \) on \( \ell \).
    Also, \( s \) was not in \( F \).
    Then, it contradicts our algorithm that \( j \) joined \( G \) on \( \ell \), as \( s \) would act defensively and join \( D \) on \( \ell \).
\end{proof}

Next, we show by \Cref{le:second-algorithm:cost-in-phase} how the cost of \textnormal{\textsc{Def}} in any but the first phase is bounded by the sizes of the sets \( G \), \( D \) and \( F \).
As we will see, we need to distinguish the servers more carefully than with the sets \( C \), \( G \), \( D \), and \( F \).
During a phase, servers that already are in \( G \) or \( D \) can transition back to \( C \) when some other server gets frozen at their current location.
For servers of \( D \), this can happen only once.
To reflect this, we split \( C \) into \( C_{1} \) and \( C_{2} \), and we split \( G \) into \( G_{1} \) and \( G_{2} \).
The sets \( C_{2} \) and \( G_{2} \) contain servers that were previously in \( D \).
When a server transitions back to \( C \), the transition itself increases the cost of the respective server to reach its final set by \( 1 \) or \( 2 \).
To capture this, we define \( e^{s}_{x} \) to be the event that server \( s \) transitions back to \( C \) and incurs an additional cost of \( x \) in the current phase.
\( E_{x} \) is the respective set of events of the current phase.
Also, among others, we introduce the following sets:
\( F_{1} \) is the set of frozen servers that are frozen at the same location as they were specifically requested before.
\( F_{1a} \subseteq F_{1} \) is the subset of these servers for which it holds that for each \( j \in F_{1a} \) there was a server \( s \in D \) at the location at which \( j \) gets frozen.
\( F_{1b} = F_{1} \setminus F_{1a} \) is the respective remaining set of servers of \( F_{1} \).
\( F_{2} = F \setminus F_{1} \) is the set of servers that get frozen on a location different from their last specific location.

\begin{restatable}{lem}{SecondAlgorithmCostInPhase}\label{le:second-algorithm:cost-in-phase}
    In any phase \( i > 1 \), the cost of \textnormal{\textsc{Def}} is
    \begin{align*}
        c^{i} \leq |\widehat{G}| + 2\,\left(|\widehat{D}| + |\widehat{F_{1}}| + |\widehat{F_{2}}| \right) + |\widehat{F_{2}}| + |E_{1}| + 2 \, |E_{2}|.
    \end{align*}
\end{restatable}

\begin{proof}[Proof of \Cref{le:second-algorithm:cost-in-phase}]
    To prove this, we consider each time step in a phase of the algorithm.
    First, consider any time step in which either a general request on a position that is covered by \textnormal{\textsc{Def}} by \( G \cup D \cup F \) happens, or a specific request for a frozen server at its position appears.
    We exclude all such time steps from the phase in the following because the algorithm takes no action in them.

    \paragraph{An Overview on all Actions.}
    From now on, for any time step, we analyze the cost and how servers move between \( G \), \( D \), \( F_{1} \), and \( F_{2} \).
    We denote by an arrow that a server leaves a set and joins another set as here: \( j: G \rightarrow D \) (\( j \) leaves \( G \) and joins \( D \)).
    We call such a movement of one server between sets a \emph{transition}.
    Note that \( C \) is always given as all servers that are not in \( G \cup D \cup F \).

    Observe that in the case that a request requires no movement of a server \( j \in C \), there is no cost.
    Thus, there are transitions \( j: C \rightarrow G \), \( j: C \rightarrow D \), \( j: C \rightarrow F_{1} \) and \( j: C \rightarrow F_{2} \) of cost zero.
    Next, we consider only the remaining cases.

    There are two kinds of time steps depending on the current request.
    Case (1) are time steps with a specific request, and case (2) are time steps without one.

    Consider case (1).
    There are two kinds of such a time step.
    Either (1.a), a server ends up in \( F_{1} \) after the time step or (1.b) it ends up in \( F_{2} \).
    When a general request is simulated, during the cases of (1), there can be additional transitions and cost exactly as in time steps of kind (2) below.

    In the case of type (1.a), server \( j \) is specifically requested at \( p^{*}(j) \).
    If (1.a.1) \( j \in D \), this incurs no cost and \( j:D \rightarrow F_{1} \).
    If (1.a.2) \( j \in G \), we have a cost of \( 1 \) for \( j: G \rightarrow F_{1} \).
    In case (1.a.2.a), there is no server of \( D \) on \( r \), else (1.a.2.b) there is \( s: D \rightarrow C \) for a cost of zero.

    In the case of type (1.b), server \( j \) is specifically requested at some location \( r \neq p^{*}(j) \).
    If (1.b.1) \( r = p(j) \), we have cost \( 0 \) and \( j: G \rightarrow F_{2} \).
    Else, (1.b.2) \( r \neq p(j) \).
    If (1.b.2.a) \( j \in C \), either (1.b.2.a.1) there is no server \( s \in G \cup D \) on \( r \), (1.b.2.a.2) there is a server \( s \in G \) at \( r \), or (1.b.2.a.3) there is a server \( s \in D \) on \( r \).
    In any case, we have cost \( 1 \) for \( j: C \rightarrow F_{2} \).
    Also, in the case of (1.b.2.a.2) \( s: G \rightarrow C \), or in the case of (1.b.2.a.3) \( s: D \rightarrow C \).
    The missing cases are (1.b.2.b) \( j \in G \) and (1.b.2.c) \( j \in D \) combined with (1.b.2.b.1 and 1.b.2.c.1) there is no server \( s \in G \cup C \) on \( r \), (1.b.2.b.2 and 1.b.2.c.2) there is \( s \in G \) on \( r \), or (1.b.2.b.3 and 1.b.2.c.3) there is \( s \in D \) on \( r \).
    In any case of (1.b.2.b), we have a cost of \( 1 \) for \( j:G \rightarrow F_{2} \).
    In any case of (1.b.2.c), we have a cost of \( 1 \) for \( j:D \rightarrow F_{2} \).
    In the cases (1.b.2.b.2) and (1.b.2.c.2), we have \( s: G \rightarrow C \), and in the cases (1.b.2.b.3) and (1.b.2.c.3), we have \( s: D \rightarrow C \).

    Consider a time step of kind (2).
    Here, there are two types of requests that can occur:
    Either (2.a) a general request on some new location appears and we have cost of \( 1 \) with \( j: C \rightarrow G \) for some \( j \), or (2.b) a general request on a position of \( p^{*}(j) \) for some \( j \in G \cup C \) appears.
    In the case of (2.b), we have a cost of \( 1 \) for \( j:G \rightarrow D \) and an additional cost of \( 1 \) for the server \( s: C \rightarrow G \) taking \( j \)'s place.
    Note that the server \( s \) cannot join \( D \), as else, \( j \) would not have been in \( G \) at the same location.

    For better readability, consider the table below listing all cases with their respective transitions and costs.
    \begin{center}
        \begin{tabular}{l|lll}
            Case      & Transition / Cost                 &                              \\ \hline
            1.a.1     & \( j: D \rightarrow F_{1} \) / 0  &                              \\
            1.a.2.a   & \( j: G \rightarrow F_{1} \) / 1, &                              \\
            1.a.2.b   & \( j: G \rightarrow F_{1} \) / 1, & \( s: D \rightarrow C \) / 0 \\ \hline
            1.b.1     & \( j: G \rightarrow F_{2} \) / 0  &                              \\
            1.b.2.a.1 & \( j: C \rightarrow F_{2} \) / 1  &                              \\
            1.b.2.a.2 & \( j: C \rightarrow F_{2} \) / 1, & \( s: G \rightarrow C \) / 0 \\
            1.b.2.a.3 & \( j: C \rightarrow F_{2} \) / 1, & \( s: D \rightarrow C \) / 0 \\
            1.b.2.b.1 & \( j:G \rightarrow F_{2} \) / 1,  &                              \\
            1.b.2.b.2 & \( j:G \rightarrow F_{2} \) / 1,  & \( s: G \rightarrow C \) / 0 \\
            1.b.2.b.3 & \( j:G \rightarrow F_{2} \) / 1,  & \( s: D \rightarrow C \) / 0 \\
            1.b.2.c.1 & \( j:D \rightarrow F_{2} \) / 1,  &                              \\
            1.b.2.c.2 & \( j:D \rightarrow F_{2} \) / 1,  & \( s: G \rightarrow C \) / 0 \\
            1.b.2.c.3 & \( j:D \rightarrow F_{2} \) / 1,  & \( s: D \rightarrow C \) / 0 \\ \hline
            2.a       & \( j: C \rightarrow G \) / 1      &                              \\
            2.b       & \( j:G \rightarrow D \) / 1,      & \( s: C \rightarrow G \) / 1 \\
        \end{tabular}
    \end{center}

    \paragraph{Restrictions to the Actions.}
    When looking at the cases above, we notice that any server \( j \) has only limited possibilities to be moved between the sets \( C \), \( G \), \( D \), \( F_{1} \), and \( F_{2} \).
    The most obvious limitation is that \emph{no server can ever leave} \( F_{1} \) or \( F_{2} \).
    Any server in one of these sets can not incur further costs within the phase.

    Now, observe that there are only seven cases in which a server \( s \) in \( G \) or \( D \) can end up in \( C \) again; cases (1.a.2.b), (1.b.2.a.2), (1.b.2.a.3), (1.b.2.b.2), (1.b.2.b.3), (1.b.2.c.2) and (1.b.2.c.3).
    In any case, the transition of \( s \) does not incur a cost.
    In all these cases, \( s \) was at a position at which some other server \( j \) was specifically requested.
    Note, if \( s \in D \), then after this time step, the position \( p^{*}(s) \) will always be covered by \( j \) for the current phase.
    That means while \( s \) joins \( C \) again, it can no longer join \( D \).
    To reflect this, split \( C \) and \( G \) into two sets: \( C = C_{1} \cup C_{2} \) and \( G = G_{1} \cup G_{2} \).
    At the beginning of the phase, all servers are in \( C_{1} \).
    In the example above, we say \( s \) joins the separate set \( C_{2} \) from which it can transition to \( G_{2} \), but any server in \( C_{2} \cup G_{2} \) cannot transition to \( D \) any more.
    Also, in all four cases, some server \( j \) must join \( F_{1a} \cup F_{2} \).
    Thus, \emph{the total number of times a server can transition to \( C_{2} \) is bounded by the total number of servers in \( F_{1a} \cup F_{2} \) at the end of the phase}.
    Besides this, note that any transition of a server between two sets has a cost at most \( 1 \).

    Next, we consider the following graph in \Cref{figure:transition-graph-2} that depicts all possible transitions between the sets with an over-approximation of the cost of a transition as the weight of the respective edge.

    \begin{figure}[htb]
        \centering
        \includegraphics[page=1, width=0.8\textwidth, clip=true, trim= 0cm 12.9cm 15.5cm 0cm]{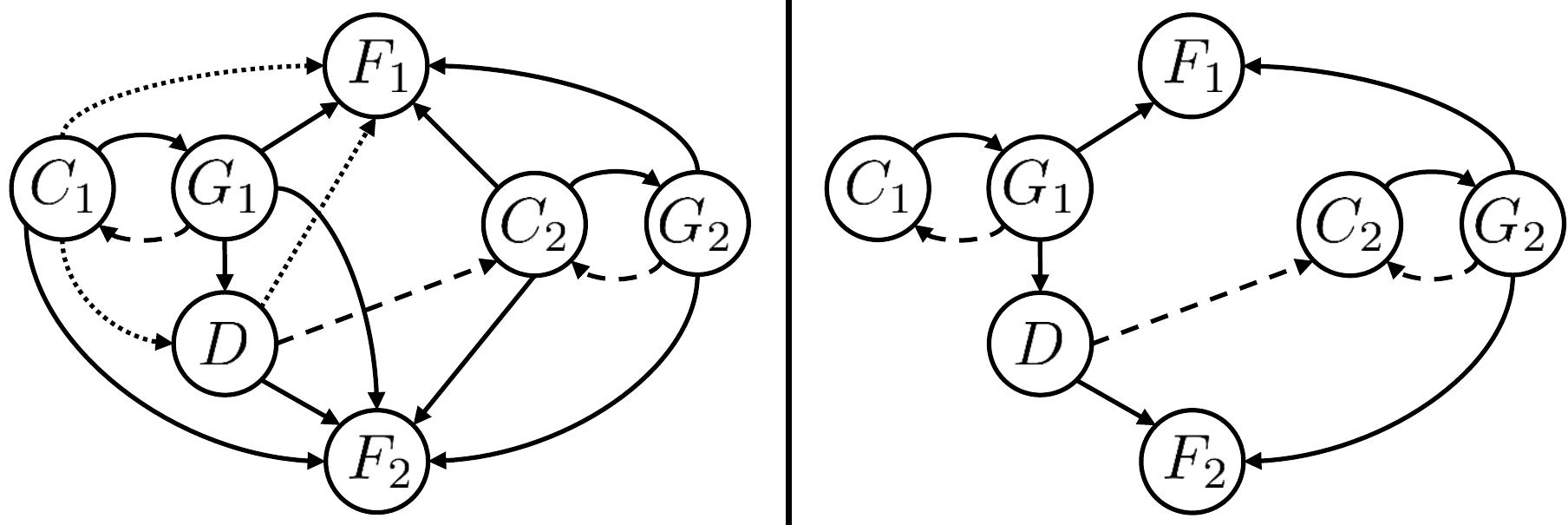}
        \caption{\textit{Left: }All possible transitions of a server within a phase.
        At the beginning of each phase all servers are in \( C_{1} \).
        Every node reflects one of the sets.
        An edge \( (U, V) \) means that a server in set \( U \) can transition to set \( V \) with a cost of at most the weight of the edge.
        Every dashed or dotted edge weighs \( 0 \), and every other edge \( 1 \).
        A dashed edge can only be traversed if another server joins \( F_{1a} \cup F_{2} \) in the same time step.
        If server \( s \) traverses \( (D,C_{2}) \) and ends up in \( \widehat{C_{2}} \cup \widehat{G_{2}} \cup \widehat{F_{1}} \), event \( e^{s}_{2} \) happens.
        Else, the crossing of a dashed edge implies that \( e^{s}_{1} \) happens.
        \textit{Right:} The transitive reduction of the graph (ignoring dashed edges) after removing the zero cost forward edges (dotted).}
        \label{figure:transition-graph-2}
    \end{figure}
    Note that the cost for a sequence of transitions of a server from set \( U \) to set \( V \) can be upper bounded by finding the longest path from \( U \) to \( V \) in the graph \( G \) of \Cref{figure:transition-graph-2}.
    Thus, to bound \textnormal{\textsc{Def}}'s cost, we can consider \( G \) without all zero cost forward edges.
    Additionally, it suffices to consider the transitive reduction of \( G \) (ignoring dashed edges).
    This simplification of \( G \) is \( G' \) depicted on the right of \Cref{figure:transition-graph-2}.

    \paragraph{On the Cost of a Phase.}
    Next, we bound the total cost of a phase.
    We argue based on \( G' \).
    First, consider all servers that end up in their final set without traversing a dashed edge (no event of \( E_{1} \) or \( E_{2} \) happens for them).
    For any such server \( j \), we have the following cost:
    If \( j \in \widehat{G_{1}} \), the cost is \( 1 \).
    If  \( j \in \widehat{D} \), or \( j \in \widehat{F_{1}} \), the cost is \( 2 \), and if \( j \in \widehat{F_{2}} \), the cost is \( 3 \).
    Second, consider the servers that end up in their final set traversing a dashed edge.
    For any such server \( j \), the cost increases by \( 2 \) only once if \( j \in \widehat{C_{2}} \cup \widehat{G_{2}} \cup \widehat{F_{1}} \) due to the crossing of the edge \( (D, C_{2}) \), i.e., an event \( e^{i}_{2} \) happens.
    Every other increase in the cost due to a dashed edge is at most \( 1 \) when an event in \( E_{1} \) happens.
    Then, the total cost of the phase \( i \) can be bounded by:
    \begin{align*}
        c^{i} & \leq |\widehat{G}| + 2 \, |\widehat{D}| + 2 \, |\widehat{F_{1}}| + 3 \, |\widehat{F_{2}}| + |E_{1}| + 2 \, |E_{2}|
    \end{align*}
\end{proof}

Next, we get rid of the set of events in the bound.
Intuitively, simplifying the bound can be achieved by a very fine-grained analysis of how events happen.
We try to bound the number of events in \( |\widehat{F_{2}}| \) as far as possible, because for each server of \( \widehat{F_{2}} \) the optimal solution must have a movement since it was lastly specifically requested.
Costs bounded by \( |\widehat{F_{2}}| \) can later be charged to these movements.

\begin{restatable}{lem}{SecondAlgorithmCostInPhaseReframed}\label{le:second-algorithm:cost-in-phase-reframed}
    In any phase \( i > 1 \), the cost of \textnormal{\textsc{Def}} is
    \begin{align*}
        c^{i} \leq 2 \, (|\widehat{C_{2}}| + |\widehat{G}| + |\widehat{D}| + |\widehat{F}|) + 5 \, |\widehat{F_{2}}| + 3 \, |\widehat{G_{2}}|.
    \end{align*}
\end{restatable}

\begin{proof}[Proof of \Cref{le:second-algorithm:cost-in-phase-reframed}]
    For the analysis of the events, we need some more notation.
    Let \( E_{2}(\widehat{F_{2}}) \subseteq E_{2} \) be the set of events of \( E_{2} \) that are triggered by a server in \( \widehat{F_{2}} \).
    Similarly, let \( E_{2}(\widehat{F_{1a}}) \subseteq E_{2} \) be the set of events of \( E_{2} \) triggered by a server in \( \widehat{F_{1a}} \).
    Then, \( E_{2} = E_{2}(\widehat{F_{2}}) \cup E_{2}(\widehat{F_{1a}}) \).
    Now, we split up \( E_{2}(\widehat{F_{1a}}) \) further:
    Consider a set \( S \in \{\widehat{C_{2}}, \widehat{G_{2}}, \widehat{F_{1b}}\} \).
    We denote by \( E_{2}(\widehat{F_{1a}}, S) \) the set of events of \( E_{2} \) triggered by a server in \( \widehat{F_{1a}} \) such that the respective server for which the event happens ends up in \( S \).
    Then, \( E_{2}(\widehat{F_{1a}}) = E_{2}(\widehat{F_{1a}}, \widehat{C_{2}}) \cup E_{2}(\widehat{F_{1a}}, \widehat{G_{2}}) \cup E_{2}(\widehat{F_{1a}}, \widehat{F_{1b}}) \).

    For each server for which an event in \( E_{2}(\widehat{F_{1a}}, \widehat{F_{1b}}) \) happens, we can find a matching server in \( \widehat{G_{2}} \cup \widehat{F_{2}} \) as follows:
    If a server \( j \in \widehat{F_{1b}} \) incurs cost of two (and an event in \( E_{2}(\widehat{F_{1a}}, \widehat{F_{1b}}) \) happens), it is in \( G_{2} \) at some location \( \ell \) just before it joins \( F_{1b} \).
    Thereafter, because \( \ell \neq p^{*}(s) \) for all \( s \) (\Cref{le:second-algorithm:on-g-no-defensive}), only a server of \( \widehat{G_{1}} \), \( \widehat{G_{2}} \), or \( \widehat{F_{2}} \) can be on \( \ell \) at the end.
    If there is a server \( s \in \widehat{G_{1}} \) on \( \ell \), \( s \) was already in \( G_{1} \) when \( j \) joined \( G_{2} \), because servers of \( C_{1} \) are preferred over servers of \( C_{2} \).
    Thus, the only way that \( s \) moves on \( \ell \) can be that it was moved back to \( C_{1} \) before due to a server of \( F_{2} \) (due to \Cref{le:second-algorithm:on-g-no-defensive} the server cannot be in \( F_{1} \)).
    In total, for server \( j \), there is a unique server \( s \in \widehat{G_{2}} \cup \widehat{F_{2}} \).
    Therefore, \( |\widehat{G_{2}}| + |\widehat{F_{2}}| \geq |E_{2}(\widehat{F_{1a}}, \widehat{F_{1b}})| \).

    Additionally, we have the following:
    For any server \( s \) the event \( e^{s}_{2} \) can happen at most once, thus for \(  S \in \{\widehat{C_{2}}, \widehat{G_{2}}, \widehat{F_{1b}}\} \) it holds \( E_{2}(\widehat{F_{1a}}, S) \leq |S| \).
    Additionally, each server of \( \widehat{F_{2}} \) triggers at most one event, and thus \( |E_{1}| + |E_{2}(\widehat{F_{2}})| \leq |\widehat{F_{2}}| \).
    Using both inequalities and the bound on \( |E_{2}(\widehat{F_{1a}}, \widehat{F_{1b}})| \), we reframe the bound of \Cref{le:second-algorithm:cost-in-phase}:
    \begin{align*}
        c^{i} & \leq |\widehat{G}| + 2 \, |\widehat{D}| + 2 \, |\widehat{F_{1}}| + 3 \, |\widehat{F_{2}}| + |E_{1}| + 2 \, |E_{2}|                                      \\
              & \leq |\widehat{G}| + 2 \,(|\widehat{D}| + |\widehat{F}|) + |\widehat{F_{2}}| + |E_{1}| + 2 \, (|E_{2}(\widehat{F_{2}})| + |E_{2}(\widehat{F_{1a}})|)    \\
              & \leq |\widehat{G}| + 2 \,(|\widehat{D}| + |\widehat{F}|) + 3 \, |\widehat{F_{2}}| + 2\, |E_{2}(\widehat{F_{1a}}, \widehat{C_{2}})|                      \\
              & \phantom{\leq} + 2\, |E_{2}(\widehat{F_{1a}}, \widehat{G_{2}})| + 2\, |E_{2}(\widehat{F_{1a}}, \widehat{F_{1b}})|                                       \\
              & \leq |\widehat{G}| + 2 \, (|\widehat{C_{2}}| + |\widehat{D}| + |\widehat{F}|) + 3 \, |\widehat{F_{2}}| + 2\, |E_{2}(\widehat{F_{1a}}, \widehat{G_{2}})| \\
              & \phantom{\leq} + 2\, |\widehat{G_{2}}| + 2\, |\widehat{F_{2}}|                                                                                          \\
              & \leq 2 \, (|\widehat{C_{2}}| + |\widehat{G}| + |\widehat{D}| + |\widehat{F}|) + 5 \, |\widehat{F_{2}}| + 3 \, |\widehat{G_{2}}|.
    \end{align*}
\end{proof}

\paragraph{On the Charging-Scheme.}
Next, we charge the cost of \textnormal{\textsc{Def}} of a phase to movement costs of \opt{}.
From now on, we denote by the exponent \( i \) the respective object of the \( i \)-th phase.
First, let us split the cost of \textnormal{\textsc{Def}} of the \( i \)-th phase into the following:
\begin{align*}
     & c^{i}_{1} = 2 \,(|\widehat{C^{i}_{2}}| + |\widehat{G^{i}}| + |\widehat{D^{i}}| + |\widehat{F^{i}_{1}}| + |\widehat{F^{i}_{2}}|) &  & c^{i}_{2} = 5 \, |\widehat{F^{i}_{2}}| \\
     & c^{i}_{3} = 3 \, |\widehat{G_{2}}|
\end{align*}

By \Cref{le:opt-cost:1-for-each-phase}, we know that \opt{} has at least one movement \( o^{i} \) for each phase \( i \) except the last one.
We charge \( c^{i}_{1} \) to \( o^{i} \) for any phase but the last.
We charge \( c^{\text{last}}_{1} \) of the last phase to a movement contributing to \( o^{1} \), because \textnormal{\textsc{Def}} has cost zero during phase \( 0 \).

Next, for any server \( j \) with respect to the current phase, let \( t^{i}_{1} \) be the last time step before phase \( i \) in which \( j \) was specifically requested and let \( p(j, t^{i}_{1}) \) be the location at which it was requested back then.
Regarding any server \( j \in \widehat{F^{i}_{2}} \), we know that \( j \)'s location at the end of the current phase is different from \( p(j, t^{i}_{1}) \) and thus, \opt{} must have moved \( j \).
We charge \( c^{i}_{2} = 5 \, |\widehat{F^{i}_{2}}| \) by charging a cost of \( 5 \) to \opt{}'s last movement of \( j \) for each \( j \in \widehat{F^{i}_{2}} \).

The charging of \( c^{i}_{3} \) is a bit more complicated.
First, we charge additional costs to the movements of servers in \( \widehat{F^{i}_{2}} \) by matching \( |\widehat{F^{i}_{2}}| \) servers of \( \widehat{G^{i}_{2}} \) to the respective movement of \opt{}.
For the remaining \( |\widehat{G^{i}_{2}}| - |\widehat{F^{i}_{2}}| \) servers, we show \Cref{le:second-algorithm:movement-mapping-G-two}.
Intuitively, for the remaining servers, we observe that our algorithm needed to move them (as they are not in \( C_{2} \), and the servers of \( C_{1} \) were already used).
As a consequence, \opt{} also needs some movement as it needs to serve the same requests.
Using the servers of \( \widehat{F_{2}} \) and \Cref{le:second-algorithm:movement-mapping-G-two}, for each server of \( G^{i}_{2} \), there is exactly one movement of \opt{} of a server \( s \) since \( s \) was lastly specifically requested before the phase until the end of the current phase.
We charge the cost of \( 3 \) to that movement, and none of these movements receives more than one charge of the current phase.
However, if a server ends up multiple times in \( \widehat{C_{2}} \cup \widehat{G_{2}} \) without being specifically requested in between, the respective movement of \opt{} could potentially receive multiple charges.
We show that in between any two times in which a server ends up in \( \widehat{C_{2}} \cup \widehat{G_{2}} \), either the server is specifically requested, or the server triggering the event is in \( \widehat{F_{2}} \) (see \Cref{le:second-algorithm:no-degrations-without-partner}).
In the former case, we charge the cost of \( 3 \) as explained above.
In the latter case, we charge the respective cost of \( 3 \) due to the later event for server \( s \) to \( j \).

\begin{restatable}{lem}{SecondAlgorithmMovementMappingGTwo}\label{le:second-algorithm:movement-mapping-G-two}
    Assume \( 0 \leq x < |\widehat{G_{2}}| - |\widehat{F_{2}}| \) servers of \( \widehat{C_{2}} \cup \widehat{G_{2}} \) were moved by \opt{} since they were lastly specifically requested until the end of a phase.
    Then, \opt{} has \( |\widehat{G_{2}}| - |\widehat{F_{2}}| - x \) movements during the phase.
\end{restatable}

\begin{proof}
    If \( x \) servers of \( \widehat{C_{2}} \cup \widehat{G_{2}} \) were moved by \opt{}, we know that \( |\widehat{C_{2}}| + |\widehat{G_{2}}| - x \) servers of \(  \widehat{C_{2}} \cup \widehat{G_{2}} \) are the entire phase on the location at which they were lastly specifically requested.
    For all these servers, another server was specifically requested at their location.
    By the algorithm, we know that there are \( |\widehat{G}| + |\widehat{D}| \) locations where only general requests appeared during the phase.
    Assume that at the end of the phase, \opt{} has \( p_{1} \geq 0 \) many servers of \( \widehat{F} \) not at the location at which they were specifically requested.
    Then the optimal solution can cover at the end \( |\widehat{C}| + |\widehat{G}| + |\widehat{D}| - |\widehat{C_{2}}| - |\widehat{G_{2}}| + x + p_{1} \) locations of \( \widehat{G} \cup \widehat{D} \).
    Therefore, the number of locations of \( \widehat{G} \cup \widehat{D} \) that are not covered by \opt{} is \( |\widehat{G_{2}}| - |\widehat{C_{1}}| - x - p_{1} \).

    Next, we show that \( |\widehat{C_{1}}| \leq |\widehat{F_{2}}| \).
    Observe that at the point in time where the first server \( j \) joins \( G_{2} \), \( C_{1} = \emptyset \).
    Else, \( j \) would not have moved because it is in \( C_{2} \) and servers of \( C_{1} \) are always preferred over servers of \( C_{2} \).
    Thus, any server \( j \in \widehat{C_{1}} \) must have been in \( G_{1} \) at location \( \ell \) before.
    Such a server can only join \( C_{1} \) again if another server \( s \) joins \( F_{2} \) on \( \ell \) (due to \Cref{le:second-algorithm:on-g-no-defensive}, the other server cannot join \( F_{1} \)).
    Therefore, \( |\widehat{C_{1}}| \leq |\widehat{F_{2}}| \).
    Using this in the above yields that there are at least \( |\widehat{G_{2}}| - |\widehat{F_{2}}| - x - p_{1} \) locations of \( \widehat{G} \cup \widehat{D} \) which are not covered at the end of the phase.

    Then, \Cref{le:opt-cost:high-when-not-all-covered} tells us that \opt{} had \( |\widehat{G_{2}}| - |\widehat{F_{2}}| - x - p_{1} + p_{1} \) movements during the phase and the lemma holds.
\end{proof}

For \Cref{le:second-algorithm:no-degrations-without-partner}, consider \Cref{figure:no-degrations-without-partner} for an intuitive depiction.
\\

\begin{minipage}[c]{0.46\textwidth}
    \centering
    \includegraphics[page=2, width=\textwidth, clip=true, trim=0cm 8cm 20cm 0cm]{figures/figures.pdf}
\end{minipage}
\begin{minipage}[c]{0.02\textwidth}
    \hfill
\end{minipage}
\begin{minipage}[c]{0.46\textwidth}
    \centering
    
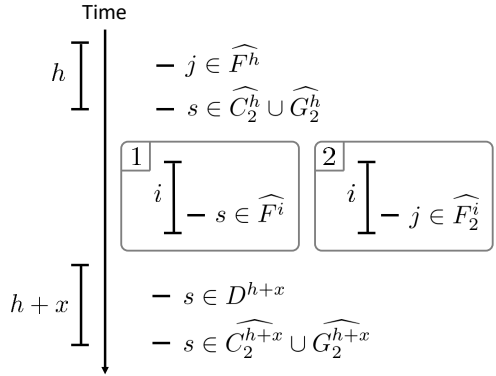
\captionof{figure}{A depiction of the statement of \Cref{le:second-algorithm:no-degrations-without-partner}.
    Time goes from the top to the bottom.
    If the server \( s \) is in \( \widehat{C^{h}_{2}} \cup \widehat{G^{h}_{2}} \) and \( \widehat{C^{h+x}_{2}} \cup \widehat{G^{h+x}_{2}} \), then either (1) in a phase \( h < i < h + x \), \( s \in \widehat{F^{i}} \), or (2) the server \( j \) triggering the event \( e^{s}_{2} \) in \( h \) is in a phase \( h < i \leq h + x \) in \( \widehat{F^{i}_{2}} \).
    Intuitively, this must be because \( s \) must have been in \( D^{h+x} \).
    This implies that \( s \) was the server which was lastly specifically requested on \( p^{*}(s) \), but after \( h \), \( j \) was lastly specifically requested on \( p^{*}(s) \).
    }
    \label{figure:no-degrations-without-partner}
\end{minipage}
\noindent\\

\begin{restatable}{lem}{SecondAlgorithmNoDegrationsWithoutPartner}\label{le:second-algorithm:no-degrations-without-partner}
    Consider a server \( s \in \widehat{C^{h}_{2}} \cup \widehat{G^{h}_{2}} \) with \( s \in \widehat{C^{h+x}_{2}} \cup \widehat{G^{h+x}_{2}} \)
    for minimal \( x > 0 \).
    Either there is a phase \( h < i < h + x \) such that \( s \in \widehat{F^{i}} \), or there there is a phase \( h < i \leq h+x \) such that \( j \in \widehat{F^{i}_{2}} \) for the server \( j \) that triggered the event \( e^{s}_{2} \) in phase \( h \).
\end{restatable}

\begin{proof}
    It holds that \( p^{*}(j) = p^{*}(s) \) at the end of phase \( h \) and \( j \) was lastly specifically requested \emph{after} \( s \).
    Since \( s \in \widehat{C^{h+x}_{2}} \cup \widehat{G^{h+x}_{2}} \), \( s \) must have been in \( D \) in phase \( h + x \).
    For this, \( s \) must be the last server for \( p^{*}(s) \) (with respect to phase \( h + x \)) that was specifically requested.
    If \( s \in \widehat{F^{i}} \) for \( h < i < h + x \) the lemma holds.
    Else, \( p^{*}(s) \) is the same for \( h \) and \( h + x \) and \( j \) must have changed its \( p^{*}(j) \) in between.
    This implies \( j \in \widehat{F^{i}} \) for \( h < i \leq h + x \).
\end{proof}

As briefly sketched above, we can now show that the maximum charges to a movement of \opt{} are limited.
Consider \Cref{figure:maximum-charges} for a depiction.
\\

\begin{minipage}[c]{0.46\textwidth}
    \centering
    \includegraphics[page=3, width=\textwidth, clip=true, trim=0cm 8cm 21.5cm 0cm]{figures/figures.pdf}
\end{minipage}
\begin{minipage}[c]{0.02\textwidth}
    \hfill
\end{minipage}
\begin{minipage}[c]{0.46\textwidth}
    \centering
    
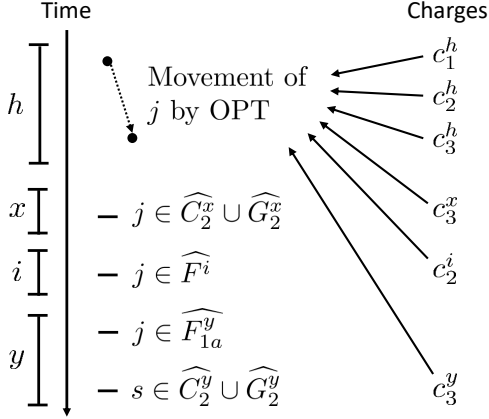
\captionof{figure}{A depiction of the situation analyzed in the proof of \Cref{le:charging:maximum-charge-per-movement-of-opt}.
    Time goes from the top to the bottom.
    A bar represents a point in time for which the adjacent statement holds.
    For details, consider the proof of \Cref{le:charging:maximum-charge-per-movement-of-opt}.
    }
    \label{figure:maximum-charges}
\end{minipage}
\noindent\\

\begin{restatable}{lem}{SecondAlgorithmMaxChargePerMovement}\label{le:charging:maximum-charge-per-movement-of-opt}
    Any movement of \opt{} gets charged at most \( 2\, (|\widehat{C}| + |\widehat{G}| + |\widehat{D}| + |\widehat{F}|) + 14 \) for some phase.
\end{restatable}

\begin{proof}[Proof of \Cref{le:charging:maximum-charge-per-movement-of-opt}]
    Consider a movement of \opt{} for server \( j \) in phase \( h \).
    Let \( i > h \) be the next phase in which \( j \) is specifically requested.

    Due to \( c^{h}_{1} \), the movement gets charged at most \( 2 \, (|\widehat{C^{h}}| + |\widehat{G^{h}}| + |\widehat{D^{h}}| + |\widehat{F^{h}}|) \).
    If \( j \in \widehat{F^{h}_{2}} \) and the movement happens before \( j \) is specifically requested, it receives a charge of at most \( 5 \) due to \( c^{h}_{2} \).
    Else, it receives charges of at most \( 5 \) due to \( c^{i}_{2} \), if \( j \) is in \( \widehat{F^{i}_{2}} \).
    From one server of \( \widehat{C^{h}_{2}} \cup \widehat{G^{h}_{2}} \) (it could also be \( j \) itself), the movement can get an additional charge of \( 3 \) (see \Cref{le:second-algorithm:movement-mapping-G-two}).
    If \( j \) ends up in \( \widehat{C^{x}_{2}} \cup \widehat{G^{x}_{2}} \) for an \( h \leq x < i \), there could be an additional charge.
    The latter charge can only be applied once because by \Cref{le:second-algorithm:no-degrations-without-partner} any second time \( j \) is in \( \widehat{C_{2}} \cup \widehat{G_{2}} \), the respective phase must be after phase \( i \).
    However, if \( j \) triggers an event \( e^{s}_{2} \) for some server \( s \) by joining \( F^{y}_{1a} \) (in phase \( y \geq i \)), and if \( s \) is in \( \widehat{C^{y}_{2}} \cup \widehat{G^{y}_{2}} \), the movement receives an additional charge of \( 3 \) due to \( c^{y}_{3} \).
    After that, \( j \) was specifically requested in \( y \) and any more charges to \( j \) affect a later movement of \( j \) by \opt{}.
\end{proof}

\paragraph{On the Competitive Ratio.}
Finally, we use that our algorithm ensures that each server is in precisely one of the sets \( C \), \( G \), \( D \), or \( F \) at any point in time, i.e., \( |C| + |G| + |D| + |F| \leq k \) always holds.

\begin{proof}[Proof of \Cref{th:defensive-algorithm:competitive-ratio}]
    Due to \Cref{le:charging:maximum-charge-per-movement-of-opt}, the algorithms cost can be charged to \opt{}'s cost such that each movement of \opt{} receives a maximum charge of \( 2\, ( |\widehat{C}| + |\widehat{G}| + |\widehat{D}| + |\widehat{F}|) + 14 \leq 2k + 14 \).
\end{proof}

\section{Algorithms for Non-uniform Metrics}\label{sec:general}

In the following section, we present two minor results on non-uniform metrics.
First, we present an algorithm for \( 2 \) servers in \Cref{sec:general:Real-Line}.
Second, we introduce an algorithm achieving a competitive ratio of \( 4 k \) on general metrics in \Cref{sec:general:WFA}.
Both algorithms work by treating each request as a general request and afterward correcting themselves if the request was specific.
While this approach did not work on uniform metrics, it gives rough upper bounds on non-uniform ones.

\subsection{The Real Line}\label{sec:general:Real-Line}

The following approach is based on the Double Coverage algorithm as presented originally for the \( k \)-server problem in \cite{chrobak_new_1990}.
\bigskip

\textbf{The Algorithm.}
Every request will first be treated as in the classical \( k \)-server problem with the DC algorithm.
If a request is a specific request for server \( j \) and server \( i \) was moved on the request, move \( i \) halfway towards \( j \) and then \( j \) on \( r \).
\bigskip

\begin{theorem}
    The algorithm above achieves a competitive ratio of \( 6 \) for \( k = 2 \) servers on the real line.
\end{theorem}

\begin{proof}
    Let \( s_{1} \) and \( s_{2} \) be the two servers of the online algorithm and let \( o_{1} \) and \( o_{2} \) be the servers of the optimal solution.
    The potential for \( k = 2 \) servers in the analysis of the Double Coverage algorithm \cite{chrobak_new_1990} is given by
    \begin{align*}
        \phi = d(s_{1}, s_{2}) + 2\, ( d(s_{1}, o_{1}) + d(s_{2}, o_{2})).
    \end{align*}

    The potential consists of the distance of the algorithm's servers to each other and the distance of the algorithm's servers to the optimal servers.
    Note, that in the original analysis, the server types do not matter and the servers are always numbered in increasing order of the line metric.
    Therefore, the second summand of \( \phi \) is better interpreted as the weight of a minimal matching \( \textnormal{\textsc{Match}} \) between the servers of the online algorithm and the optimal servers and thus,
    \begin{align*}
        \phi = d(s_{1}, s_{2}) + 2\cdot \textnormal{\textsc{Match}}.
    \end{align*}

    One can see here that the original potential does not depend on the servers' types simply because it does not need to.
    We extend the potential by relating the servers of the online algorithm with their counterpart of the optimal solution again.
    Then, we end up with the following extension and generalization of \( \phi \):
    \begin{align*}
        \psi
        = \alpha \cdot d(s_{1}, s_{2}) + \beta \cdot \textnormal{\textsc{Match}} + \gamma \cdot (d(s_{1}, o_{1}) + d(s_{2}, o_{2}))
    \end{align*}

    In the following, we determine the values for \( \alpha \), \( \beta \), and \( \gamma \) by going through all possible cases.

    First, consider the DC moves.
    Assume the algorithm moves \( s_{1} \) outwards by \( 1 \).
    Of course, the algorithm could move a server farther but we can argue on the value normalized to \( 1 \), because only the relations between \( \alpha \), \( \beta \), and \( \gamma \) in the potential matters.
    In this case, the term \( d(s_{1}, s_{2}) \) increases by 1, the matching decreases by 1, and the term \( d(s_{1}, o_{1}) \) may also increase up to 1.
    Hence, \( \Delta \psi \leq \alpha - \beta + \gamma \).
    The cost of the algorithm (\( 1 \)) is canceled if
    \begin{align}
        1 + \Delta \psi \leq 0
        \Leftarrow
        1 + \alpha - \beta + \gamma \leq 0
        \Leftrightarrow
        \alpha - \beta + \gamma \leq - 1.
        \label{ineq_Line_constraint-1}
    \end{align}

    Now consider the case where both servers move inwards, both by distance 1.
    The matching \( \textnormal{\textsc{Match}} \) remains neutral as at least one optimal server lies between the algorithm's servers.
    With the same argument, at least one of the terms \( d(s_{1}, o_{1}) \) and \( d(s_{2}, o_{2}) \) decreases by 1 as well, making the change with regard to their sum at most 0.
    Meanwhile, the distance \( d(s_{1}, s_{2}) \) decreases by 2, giving \( \Delta \psi \leq - 2 \, \alpha \).
    The cost of the algorithm (\( 2 \)) is again canceled if
    \begin{align}
        2 + \Delta \psi \leq 0
        \Leftarrow
        2 - 2 \, \alpha \leq 0
        \Leftrightarrow
        - \alpha \leq - 1.
        \label{ineq_Line_constraint-2}
    \end{align}

    Finally, we have to consider the swap move which is performed if the wrong server is on the request after the Double Coverage move.
    Consider the following setup:
    The server \( s_{1} \) is at the location of the request, but server \( s_{2} \) is needed.
    The server \( o_{2} \) of the optimal solution is on the request.

    Now, \( s_{2} \) moves distance 2 onto the request while \( s_{1} \) moves distance 1 in the opposite direction in which \( s_{2} \) moves.
    We can map the locations onto a number line as follows: The request is at 0 and \( s_{2} \) is at 2.
    During the move, \( s_{2} \) moves towards 0 and \( s_{1} \) towards 1.
    Going at equal speed, both servers arrive at 1 at the same time, where \( s_{1} \) stops and \( s_{2} \) continues to go to 0.

    In any case, we can see that \( d(s_{1}, s_{2}) \) decreases by 1.
    The rest of the potential change now depends on the location of \( o_{1} \).
    First case: \( o_{1} \) is on or to the right of location 1.
    This means \( s_{1} \) moves towards \( o_{1} \) the entire time. Since \( s_{2} \) moves onto the location of \( o_{2} \), we get a total decrease of 3 in the term \( d(s_{1}, o_{1}) + d(s_{2}, o_{2}) \).
    The change in \( \textnormal{\textsc{Match}} \) can be best observed from the perspective of \( o_{1} \): First, a server moves away from it by at most 2.
    Then a server moves towards it by distance 1.
    There is no difference in the potential involving \( o_{2} \) as there is a server at its location at the beginning and the end.
    Therefore \textnormal{\textsc{Match}} increases by at most 1.
    In total, \( \Delta \psi \leq -\alpha + \beta -3 \gamma \).
    The cost of the algorithm (\( 3 \)) is canceled by the potential if
    \begin{align}
        3 + \Delta \psi \leq 0
        \Leftarrow
        3 - \alpha + \beta -3 \gamma \leq 0
        \Leftrightarrow
        - \alpha + \beta -3 \gamma \leq - 3.
        \label{ineq_Line_constraint-3}
    \end{align}

    Second case: \( o_{1} \) is to the left of 1.
    Now \( d(s_{1}, o_{1}) \) increases by up to 1 while \( d(s_{2}, o_{2}) \) again decreases by 2.
    The change in the matching is as follows:
    If \( o_{1} \) is to the left of 0, then \( s_{1} \) increases the distance towards both servers by 1 while \( s_{2} \) decreases it by 2, making an overall decrease by 1.
    If \( o_{1} \) is between 0 and 1, observer that when the servers \( s_{1} \) and \( s_{2} \) meet at 1, they switch the partners in \textnormal{\textsc{Match}}.
    Hence, \( s_{2} \) moves towards its matching partner the entire time.
    Overall we have \( \Delta \psi \leq - \alpha - \beta - \gamma \) and the cost of the algorithm (\( 3 \)) is canceled by the potential if
    \begin{align}
        3 + \Delta \psi \leq 0
        \Leftarrow
        3 - \alpha - \beta - \gamma \leq 0
        \Leftrightarrow
        - \alpha - \beta - \gamma \leq - 3.
        \label{ineq_Line_constraint-4}
    \end{align}

    Whenever \opt{} moves its servers, the potential increases by at most \( (\beta + \gamma) \) times the moved distance, because the first term is independent of \opt{}.
    Choosing \( \alpha = 1 \), \( \beta = 4 \) and \( \gamma = 2 \) ensures that \Cref{ineq_Line_constraint-1,ineq_Line_constraint-2,ineq_Line_constraint-3,ineq_Line_constraint-4} hold while \( (\beta + \gamma) \) is minimized.
    Since the increase in the potential is upper bounded by \( (\beta + \gamma) = 6 \), the competitive ratio is at most \( 6 \).
\end{proof}

\subsection{General Metrics}\label{sec:general:WFA}

The following algorithm is based on the Work Function Algorithm~\cite{koutsoupias_k-server_1995}.
\bigskip

\textbf{The algorithm.}
Upon the arrival of a request, treat it as a general request and execute a step of the Work Function algorithm~\cite{koutsoupias_k-server_1995}.
If the request is specific to a server, move that server to the request afterward.
\bigskip

\begin{theorem}
    The above algorithm achieves a competitive ratio of at most \( 4 \, k \) on general metrics.
\end{theorem}

\begin{proof}
    Let \alg{} be the algorithm above.
    Using the analysis in~\cite{koutsoupias_k-server_1995}, we get as in the proof of Theorem 4.3 of~\cite{koutsoupias_k-server_1995} that \( C(\opt{}) + C(\alg{})^{\textnormal{WFA}} \) regarding only the Work Function movements is bounded by the potential difference in each step and sums up to at most \( \sum \Delta \phi \leq 2k\cdot C(\opt{}) \) in total (plus a constant dependent on the initial configuration that we ignore in the following).
    Let \( \phi \) be the potential from that analysis.
    Now we introduce an additional potential \( \psi=\sum_{i=1}^{k} d(s_{i}, o_{i}) \) that reflects the sum over the distances between equivalent servers \( s_{i} \) of \alg{} and \( o_{i} \) of \opt{}.
    Note, how \( \psi \) is at least zero at any time.

    Let \( P \) be the cost to move a server of the algorithm towards a request if it is specifically for that server.
    \( \psi \) decreases by \( P \) since the server must be at the location of the request in \alg{} and \opt{}.
    Therefore, \( \Delta \psi^{\textnormal{Spec}} \) cancels the cost \( C(\alg{})^{\textnormal{WFA}} \) of \alg{} during such a move.
    During the regular Work Function move, \( \psi \) increases by at most \( \Delta \psi^{\textnormal{WFA}} \leq C(\opt{}) + C(\alg{})^{\textnormal{WFA}} \) with regard to that move by definition.
    Therefore, the total cost of \alg{} is at most
    \begin{align*}
        C(\alg{})
        & \leq C(\alg{}) + \sum \Delta \psi
        \leq
        C(\alg{})^{\textnormal{WFA}} + \sum \Delta \psi^{\textnormal{WFA}}
        \\
         & \leq C(\alg{})^{\textnormal{WFA}} + C(\opt{}) + C(\alg{})^{\textnormal{WFA}}
        = 2 \, C(\alg{})^{\textnormal{WFA}} + C(\opt{})
        \\
         & \leq 2 \sum \Delta \phi
        \leq 4 \, k \cdot C(\opt{}).
    \end{align*}
\end{proof}

\section{Closing Remarks}\label{sec:closing-remarks}

In this paper, we introduced the \( k \)-Server with Preferences Problem, a generalization of the \( k \)-Server Problem, that poses new difficulties already on uniform metrics.
While the competitive ratio is, by definition, a worst-case ratio, our results show a parameterization that connects the different input types of our generalization in all its forms (classical \( k \)-Server inputs vs. mixed inputs vs. trivial inputs).

It would greatly complement our work to have a solution for the offline problem.
Already determining the complexity of solving the problem offline seems to be challenging.

We assumed that the initial configuration of the online algorithm was identical to the optimal one.
This assumption can be dropped without problems, increasing the competitive ratio of our algorithms by an additive term independent of the optimal solution.
All of the lower bounds end up in a configuration equivalent to the initial one.
Hence, their sequence can repeat arbitrarily often, such that the same bounds hold in the limit.

An important question for future work is if and how our results extend to other metric spaces.
Presumably, for the case of general metric spaces, the severity of a mismatch of the servers after the optimal positions are covered is increased.
We analyzed how the competitive ratio increases if an algorithm neglects to keep its servers close to their last known configuration (revealed and rendered significant by specific requests).
On a uniform metric, the influence is significant, even though each server is at most a distance \( 1 \) from its position in the optimal solution.
From a technical perspective, the case of a general metric probably needs a creative way for an algorithm to measure its deviation from the last (partially known) configuration of the optimum.
As we have seen, an adaption for the Work-Function-Algorithm that treats every request as a general one and, afterward, corrects itself if the request was specific, yields a competitive ratio of \( \mathcal{O}(k) \).
However, based on our observations for uniform metrics, we believe that more involved techniques are required to come close to the lower bound.
Probably there is a similar trade-off between instances of the \( k \)-Server Problem and instances including specific requests.

Regarding the model, there are various possible extensions due to the heterogeneity.
For example, one could think that requests bring forth a set of servers of which not all but a subset is needed.
This makes the problem even more difficult for an online algorithm.
It is also interesting to see if further trade-offs influencing the competitive ratio can be determined.

\bibliography{references}

\end{document}